\documentclass[%
twocolumn,
 amsmath,amssymb,
 aps, physrev,
superscriptaddress
]{revtex4-2}

\usepackage[utf8]{inputenc}
\usepackage[T1]{fontenc}
\usepackage{microtype}
\usepackage{lmodern}
\usepackage{textcomp}
\usepackage{amssymb,amsthm}
\usepackage{mathtools}

\newtheorem{lemma}{Lemma}[section]

\usepackage{graphicx}% Include figure files
\usepackage{color}

\usepackage{siunitx}
\usepackage{hyperref}
\usepackage{orcidlink}

\begin{document}

%\preprint{APS/123-QED}

\title{\textbf{Big Crunch on Julia and Anti-Julia Sets in the Integrable Limit} 
}% 

\author{Masatomi Iizawa\,\orcidlink{0000-0002-3735-0616}}
\email{Contact author: masatomi.iizawa@tu-braunschweig.de}
\affiliation{Institut f\"ur Theoretische Physik, Technische Universit\"at Braunschweig, Mendelssohnstr. 3, D-38106 Braunschweig, Germany}
% https://orcid.org/0000-0002-3735-0616

\author{Satoru Saito\,\orcidlink{0009-0000-7676-5032}}
\email{saito_ru@nifty.com}
\affiliation{Liberality Research, 5-22-6 Matsubara, Setagaya, Tokyo 156-0043, Japan}
% https://orcid.org/0009-0000-7676-5032

\author{Yasuhito Narita\,\orcidlink{0000-0002-5332-8881}}%
\email{y.narita@tu-braunschweig.de}
\affiliation{Institut f\"ur Theoretische Physik, Technische Universit\"at Braunschweig, Mendelssohnstr. 3, D-38106 Braunschweig, Germany}
\affiliation{Max Planck Institute for Solar System Research, Justus-von-Liebig-Weg 3, D-37077 G\"ottingen, Germany}
% https://orcid.org/0000-0002-5332-8881

\author{Yasuyuki Maeda\,\orcidlink{0009-0000-4187-9945}}
\affiliation{Liberality Research, 5-22-6 Matsubara, Setagaya, Tokyo 156-0043, Japan}
% https://orcid.org/0009-0000-4187-9945

\author{Hiromitsu Harada\,\orcidlink{0009-0006-5467-7131}}
\affiliation{Liberality Research, 5-22-6 Matsubara, Setagaya, Tokyo 156-0043, Japan}
% https://orcid.org/0009-0006-5467-7131

\author{Akira Shudo\,\orcidlink{0000-0001-9443-9054}}
\email{shudo@tmu.ac.jp}
\affiliation{Department of Physics, Faculty of Science, Tokyo Metropolitan University, 1-1 Minami-Osawa, Hachioji-shi, Tokyo 192-0397, Japan}
% https://orcid.org/0000-0001-9443-9054


\begin{abstract}
The Julia set is defined by the closure of repelling periodic points in chaotic systems. Why does this structure not appear in integrable systems? In this paper, we address this question by demonstrating the existence of the closure of divergences of the periodic equations, which we designate as the ``anti-Julia set.'' We also call the sets before taking the closures the pre-Julia and pre-anti-Julia sets, respectively. We illustrate the transition mechanism by considering a complex map that interpolates between integrable and non-integrable dynamics, by introducing a real deformation parameter $a$. 
For $0<a\le 1/2$, we show that the Julia set and the anti-Julia set coincide in the complex plane, although the  pre-Julia and pre-anti-Julia sets  remain completely disjoint. At the integrable limit $a\to 0$, these two dual structures undergo a critical collision and subsequent annihilation, reminiscent of a cosmological \textit{Big Crunch}. When $1/2<a<1$, on the other hand,  the boundary of the pre-anti-Julia set includes the Julia set. We analytically characterize these phenomena, focusing in particular on the asymptotic behavior of the pre-anti-Julia set as it approaches the integrable limit, and provide numerical visualizations that elucidate the underlying mechanisms of this \textit{Big Crunch} phenomenon.
\end{abstract}

\maketitle

%\tableofcontents

%%
%% sec. 1
%%

\section{Introduction}\label{sec:introduction}
The study of nonlinear dynamics has long been polarized into two distinct paradigms: integrable systems and chaotic systems. The community surrounding integrable systems has focused on uncovering hidden symmetries, constructing exact solutions, and analyzing preserved structures. This field was catalyzed by the discovery of solitons in the Korteweg--de Vries (KdV) equation by Zabusky and Kruskal (1965)~\cite{Zabusky1965}, which led to the development of the Inverse Scattering Method by Gardner, Greene, Kruskal, and Miura (1967)~\cite{Gardner1967} (see a review by Zabusky~\cite{Zabusky2005}). Subsequent algebraic formulations, initiated by Sato’s discovery of the infinite-dimensional Grassmann manifold structure underlying soliton equations~\cite{Sato1981, Sato1983} and further developed through the introduction of $\tau$-functions by the Kyoto school~\cite{Jimbo1981_PhysD_1,Jimbo1981_PhyD_2,Jimbo1981_PhyD_3}, elevated the field into a framework of mathematical physics based on infinite-dimensional Lie algebras~\cite{Kashiwara1981,Date1981_PJA,Date1981_JPSJ1,Date1982_PhysD,Date1982_RIMS1,Date1981_JPSJ2,Date1982_RIMS2}.

Conversely, the chaos and pattern-formation community has investigated predictability limits, statistical properties, and geometric structures in non-integrable systems. This paradigm was initiated by Turing's seminal work on reaction--diffusion systems (1952)~\cite{Turing1952} and Lorenz's identification of deterministic non-periodic flow (1963)~\cite{Lorenz1963}. The field was further conceptualized by Ruelle and Takens (1971)~\cite{Ruelle1971}, who introduced the concept of strange attractors in fluid turbulence. The discovery of universal scaling behavior in period-doubling bifurcations by Feigenbaum (1978)~\cite{Feigenbaum1978} established that deterministic chaos possesses structural universality independent of the specific system details. Since the notion of J-stability was introduced by Ma\~n\'e \textit{et al.} (1983)~\cite{Mane1983}, much effort has been devoted to clarifying the structural stability of the Julia set, and this line of research continues today (e.g. Astorg \textit{et al.} (2021--2023)~\cite{Astorg2021}).

Because of these differences in mathematical machinery (algebraic and analytical for integrability, versus geometric and topological for chaos), these two fields have long been pursued by largely isolated research communities. Investigations into near-integrable systems, such as the Kolmogorov-Arnold-Moser (KAM) theory~\cite{Kolmogorov1954,Kolmogorov1979,Arnold1963_ru,Arnold1963,Moser1962} (see also well-known text~\cite{Lichtenberg1992, Arnold2006}), have explored this boundary. The KAM theory establishes the local persistence of integrable structures under small perturbations. However, these studies focus primarily on Hamiltonian systems and are limited to the local behavior near the integrable state. They do not describe the global transition from fully developed chaos to complete integrability in general dynamical systems.

Despite this historical divide, significant efforts have been made to bridge the gap between chaos and integrable systems by investigating the continuous transition between these two regimes. The foundation for understanding discrete integrable systems, which act as the counterpart to deterministic chaos, was laid through extensive studies on Hirota's bilinear difference equation (also known as the Hirota--Miwa equation)~\cite{Saito2012}. It has served as a master equation for deriving a wide range of integrable models. With this robust framework of discrete integrability established---and driven by the realization that restricting our focus solely to integrable systems yields an inherently incomplete picture, thereby losing a deeper understanding of integrability itself---investigations turned to interpolating between integrable and non‑integrable regimes from a global perspective within the same unified setting. Initial investigations analyzed a family of rational maps interpolating integrable and non-integrable difference analogues of the logistic equation~\cite{Saitoh1995,Saitoh1996}. These studies observed how specific discretizations preserved integrability (sharing the symmetries of the KP hierarchy) and bypassed the formation of the Julia sets typical of chaotic maps. This breakthrough revealed that the true nature of integrability cannot be fully grasped without exploring this boundary.

Subsequent studies on higher-dimensional rational maps, including the discrete Euler top, uncovered a geometrical distinction regarding periodic points~\cite{Saito2006,Saito2007_JPSJ}. Whereas periodic points in non-integrable maps are generally isolated, periodic points in integrable systems possessing a sufficient number of invariants form continuous algebraic varieties, termed Invariant Varieties of Periodic Points (IVPPs)~\cite{Saito2006,Saito2007_JPSJ}. The relation between IVPPs and recurrence equations was formalized~\cite{Saito2007_JPhysA}, and perturbative studies further revealed how these continuous IVPPs break down into isolated points when integrability is broken~\cite{Saitoh2008}. To elucidate the precise mechanism of the integrable-nonintegrable transition, the \textit{fate of the Julia set} at the integrable limit has been extensively explored~\cite{Saito2010,Saito2013}. Analytical investigations revealed that, as a map approaches integrability, the isolated repelling periodic points (repeller; repulsive periodic points) dense in the Julia set do not merely disappear. Instead, the Julia set approaches a Variety of Singular Points (VSPs) and degenerates into  Indeterminate Points (or Indefinite Points; IDPs) along algebraic curves, acting as singular loci~\cite{Saito2010,Saito2013}. Furthermore, the phenomenon of singularity confinement, which is often used to characterize discrete integrable systems, was shown to be intimately related to the generation of IVPPs and could be interpreted within the elegant framework of projective resolutions in triangulated categories~\cite{Saito2014,Yumibayashi2014}.

While the aforementioned studies extensively investigated the fate of the pre-Julia set defined by the repelling periodic points of the map, the dual structure associated with the divergent points, also called pre-poles of the map remains largely unexplored. Conversely, the map possesses an equal number of divergent points, defined by the poles of the periodic equation, which aggregate into what we designate as the \textit{pre-anti-Julia set}. As a system approaches the integrable limit, these two dual structures, the pre-Julia set and the pre-anti-Julia set, do not evolve independently. Instead, we demonstrate that at the integrable limit, the pre-Julia set and the pre-anti-Julia set(periodic-points and pre-pole families) undergo a critical collision and subsequent annihilation, a dynamic process reminiscent of a cosmological \textit{Big Crunch}. This catastrophic event is not merely a geometric curiosity; it provides a definitive mechanism for the integrable--nonintegrable transition, revealing how chaotic invariant sets degenerate into the singular loci and IDPs that characterize discrete integrable systems. In this paper, we analytically characterize the asymptotic properties of the pre-Julia and pre-anti-Julia sets approaching the integrable limit. Furthermore, we provide numerical visualizations to elucidate the mechanisms of this \textit{Big Crunch} phenomenon.

The remainder of this paper is organized as follows. In Sec.~\ref{sec:duality}, we introduce the non-integrable deformation from the Möbius map to the logistic map. By reviewing the degeneration of the Julia set towards the integrable limit, we establish the necessity of investigating the dual divergent structures. Sec.~\ref{sec:antijulia} formally defines the anti-Julia set via the poles of the periodic equation, detailing its algebraic and geometric structures along with its asymptotic behavior. In Sec.~\ref{sec:bigcrunch}, we highlight the \textit{Big Crunch} phenomenon at the integrable limit. Here, we analytically characterize the critical collision of the dual structures and the remnants of their annihilation, such as singular loci and IDPs, complemented by numerical visualizations. Finally, Sec.~\ref{sec:conclusion} concludes the paper by providing a unified picture of the integrable--nonintegrable transition.

%%
%% sec. 2
%%
\section{Duality in Deformed Integrable Maps: Periodic Points and Pre-Poles}\label{sec:duality}
\subsection{Non-Integrable Deformations from the Möbius Map to the Logistic Map}
Every holomorphic automorphism on the Riemann sphere $\hat{\mathbb{C}} = \mathbb{C} \cup \{\infty\}$ is uniquely given by a Möbius map $z \mapsto \frac{\alpha z+\beta}{\gamma z+\delta}\quad (\alpha \delta-\beta\gamma\neq 0)$, and the automorphism group $\text{Aut}(\hat{\mathbb{C}})$, a set of all Möbius maps, is identical to the entire class of one-dimensional integrable systems. While iterations of a Möbius map preserve generalized circles, non-invertible rational maps of degree two or higher invariably destroy this geometrical simplicity, giving rise to chaotic dynamics. The minimal step to break integrability is thus to introduce a degree-two deformation. Because any quadratic polynomial can be conjugated into the logistic map $z \mapsto \mu z (1-z)$ via a Möbius transformation, the logistic map is the canonical form representing the entire class of degree-two polynomial dynamics.

To study the deformation between the Möbius map and the logistic map via the deformation parameter $a$, we construct a linear homotopy. The resulting deformed map provides a geometric framework to examine the boundary between integrability and chaos. If we demand that this deformed map satisfies the following criteria:
\begin{enumerate}
    \item it reduces to the integrable Möbius map at a deformation parameter $a=0$,
    \item it exactly coincides with the chaotic logistic map at $a=1$, and
    \item it preserves the fixed points at $0$ and $1-1/\mu$, as well as the multiplier $\mu$ at the origin, independently of the parameter $a$,
\end{enumerate}
the resulting rational map is uniquely determined. This map, termed the Deformed Logistic Map (DLM), is given by:
\begin{align}
z \mapsto f_{a} (z) = \mu z \frac{1-az}{1+\mu(1-a)z}. \label{eq:dlm} \end{align}
The DLM was introduced in the context of KP hierarchy by Saitoh et al. (1995)~\cite{Saitoh1995} and subsequently analyzed in detail (1996)~\cite{Saitoh1996}. The uniqueness of the DLM makes it the inevitable choice for investigating the integrable--nonintegrable transition.

\subsection{Degeneration of the Julia Set towards the Integrable Limit}
The Julia set is defined as the closure of the relulsive periodic points of a rational map, which can equivalently be constructed as the closure of all backward iterates starting from an arbitrary repelling periodic point~\cite{Devaney1989}. As long as the deformation parameter $a$ is non-zero, DLM remains non-integrable, and thus the Julia set necessarily exists. Since the periodicity conditions are algebraic functions of the dynamical variables and $a$ for any period $n$, the algebraic nature of the set is preserved across all non-zero values of $a$. For the one-dimensional DLM, it has been proven that, when $a$ is less than the critical value 
$a^*=\frac{\sqrt\mu}{2(1+\sqrt\mu)}$, the Julia set is entirely confined to the real axis, if the bifurcation parameter $\mu$ is real and reduces to a totally disconnected Cantor set with an empty interior~\cite{Saitoh1996}. This topological property provides the constraint necessary to evaluate the asymptotic spatial distribution of the chaotic invariant set just before its degeneration. Furthermore, numerical calculations of the information and box dimensions using modest computational resources (typical of PCs over three decades ago) suggest that these fractal dimensions monotonically decrease as $a\to 0$, pointing to a trend where the Julia set becomes spatially dilute and systematically shrinks on the real axis prior to its ultimate degeneration~\cite{Saitoh1996}. While a detailed analysis of these dimensions remains a promising topic for future work, the evaluation of numerical errors that we will present in subsequent sections provides a solid foundation for such dimensional investigations.

As noted in the previous section, in higher-dimensional maps, the Julia set approaches the IDP in the integrable limit. Since our DLM is one-dimensional, there is no IVPP. Instead, all points of period $n$ cover the complex $z$-plane only when the parameter $\mu$ takes the value $e^{im\pi/n},\ m=1,2,\cdots, n-1$. Otherwise, as long as $a=0$, no periodic points exist. In \S4, we will show that the Julia set of the DML also degenerates to the IDP in the limit $a\to 0$.

The degeneration of the Julia set into these singular loci provides a link to the algebraic structure of discrete integrable systems. 
The phenomenon of singularity confinement is defined as a dynamical process in which an orbit of a rational map falls into a singularity and diverges to infinity, but returns to a finite point after a finite number of iterations~\cite{Yumibayashi2014, Saito2014}. It was demonstrated that when a map restarts from the singular loci (IDPs and VSPs) into which the Julia set has degenerated, the process of singularity confinement iteratively generates the continuous IVPPs of all periods~\cite{Saito2010, Yumibayashi2014}. Furthermore, this generation of IVPPs out of singularities is formulated as a projective resolution within the framework of triangulated categories~\cite{Saito2014}. Consequently, the degeneration of the chaotic invariant set into IDPs, namely VSPs constitutes the universal mechanism of the integrable--nonintegrable transition. In other words the pre-anti-Julia set becomes IDPs in the integrable limit.

\subsection{Necessity of the Dual Perspective}\label{subsec:necessity_dual}
As mentioned in the preceding subsection, the isolated repelling periodic points dense in the Julia set ultimately degenerate into singular loci, specifically IDPs. By definition, an IDP of a rational map is a singularity where both the numerator and the denominator vanish simultaneously. This fact implies that the degeneration of the pre-Julia set cannot be an isolated phenomenon; it requires the structural involvement of the poles of the map. Indeed, we show in Appendix~\ref{appendix_cancellation} that the denominators of the periodicity conditions possess the same number of zeros as the numerators. These zeros cancel only when $a=0$. This makes it clear that the divergence of the periodicity conditions is responsible for the disappearance of the pre-Julia set in the integrable limit, $a \to 0$.

In the context of the $n$-th periodic equation, $f_a^{(n)}(z)-z=0$, the roots of the numerator define the conventional periodic points. Conversely, the roots of the denominator define the pre-poles, which are mapped to infinity. In non-integrable rational maps, these divergent points aggregate to form a complementary structure, which we designate as the \textit{pre-anti-Julia set}. Because the IDPs emerge where the finite periodic points and the pre-poles coalesce, tracking the degeneration of the conventional pre-Julia set demands the simultaneous tracking of these divergent points. This internal duality, between the pre-Julia and pre-anti‑Julia sets (finite periodic points and the pre-poles) within the same rational map, provides the missing key to understanding the integrable--nonintegrable transition. The pre‑poles play an important role in determining the stability properties of the Julia set~\cite{Beardon1991}. Since our focus here is on its dual relationship with the pre-Julia set, we refer to it as the pre-anti‑Julia set in this article.

%%
%% sec. 3
%%
\section{The Anti-Julia Set: Dynamics of the  Pre-Poles}\label{sec:antijulia}
\subsection{Defining Pre-Poles via Poles of the Periodic Equation}
The $n$-th iterate $f_a^{(n)}(z)$ of the DLM $f_a(z)$ is a rational function. By isolating the trivial contribution at $z=0$, the map is represented in terms of coprime polynomial factors as
\begin{align}
f_a^{(n)}(z)=\mu^n z\frac{N_a^{(n)}(z)}{D_a^{(n)}(z)}\quad \xrightarrow[a \to 0]{}\quad f_0^{(n)}(z)=\frac{\mu^n z}{1+\mu\frac{1-\mu^n}{1-\mu}z}
\label{fantof0n}
\end{align}
where $N_a^{(n)}(z)$ and $D_a^{(n)}(z)$ encode the non-trivial zeros and finite poles, respectively. The conventional periodic points of period $n$ are given by the roots of $f_a^{(n)}(z) - z = 0$. For non-zero points, clearing the denominator yields the numerator condition $\mu^n N_a^{(n)}(z) - D_a^{(n)}(z) = 0$. Since this polynomial always contains the factor $X(z) = \mu z - \mu + 1$ associated with the fixed point (period-1 point), we factor it out by defining $A_a^{(n)}(z)$ via
\begin{align}
X(z)A_a^{(n)}(z) \coloneqq \mu^n N_a^{(n)}(z) - D_a^{(n)}(z).
\label{XA}
\end{align}
Here, $A_a^{(n)}(z)$ determines nontrivial solutions of $f_a^{(n)}(z)=z$ after removing the fixed-point factor, for $n\ge 2$. Consequently, the periodicity condition translates to the rational equation:
\begin{align}
\frac{A_a^{(n)}(z)}{D_a^{(n)}(z)} = 0,\qquad n=2,3,4,\cdots.
\label{A/D}
\end{align}
The finite roots of this equation are determined by the numerator equation $A_a^{(n)} (z) = 0$. The Julia set is formed by the closure of the repelling periodic points derived from these numerator roots.

Conversely, the poles of the periodic equation are determined by the roots of the denominator:
\begin{align}
D_a^{(n)}(z) = 0.
\end{align}
A root of this equation represents a pre-pole that maps to infinity under $n$ iterations, i.e., a preimage $f_a^{(-n)}(\infty)$ of $\infty$. While the dynamics of the numerator $A_a^{(n)}(z)$ determine the finite periodic structure of the map, the dynamics of the denominator $D_a^{(n)}(z)$ determine its divergent structure.

For $a\ne 0$, the polynomials $A_a^{(n)}(z)$ and $D_a^{(n)}(z)$ share no common roots, meaning that the finite periodic points and the pre-poles are distinct and completely separated in the phase space except at $\infty$. As $n$ increases, the pre-poles aggregate to form a set. We define the pre-anti-Julia set by this set and the anti-Julia set by its closure. The anti-Julia set constitutes the dual structure to the Julia set, originating entirely
from the divergence of the periodicity conditions

\subsection{Coexistence of the Julia and Anti-Julia Sets as Cantor Sets}
For a non-zero deformation parameter, $a$, the polynomials $A_a^{(n)} (z)$ and $D_a^{(n)}(z)$ are relatively prime and share no common roots. 
Thus, if we define the set of finite periodic points $P_a^{(n)}$ and the set of pre-poles $Q_a^{(n)}$ of period $n$ as
\begin{align}
P_a^{(n)} &= \{ z \in \mathbb{C} \mid A_a^{(n)}(z) = 0 \}, \\
Q_a^{(n)} &= \{ z \in \mathbb{C} \mid D_a^{(n)}(z) = 0 \},
\end{align}
this algebraic fact guarantees that they are completely distinct:
\begin{align}
P_a^{(n)} \cap Q_a^{(n)} = \emptyset \quad (a \neq 0).
\end{align}
Consequently, the phase space accommodates two separate infinite discrete structures: the (conventional) pre-Julia set formed by the periodic points, and the pre-anti-Julia set formed by the pre-poles.

As stated in Sec.~\ref{sec:duality}, there is a specific deformation parameter regime where the conventional Julia set is entirely confined to the real axis and forms a Cantor set~\cite{Saitoh1996}. In this same regime, all roots of the denominator equation $D_a^{(n)}(z)=0$ are also real. Because these divergent poles are generated by the backward iterations of the same real rational map, the closure of the pre-poles forms the same Cantor Julia set on the real axis. This observation admits a broader formulation. It is known that when $\infty$ is a repelling fixed point of a map $f$, the closure of the set of all pre‑poles,
\begin{align}
J_{\mathrm{anti}}(f);=\overline{\bigcup_{n\ge 1} f^{(-n)}(\infty)}
\label{antiJulia}
\end{align}
coincides with the Julia set $J(f)$. For the DLM $f_a(z)$, the repelling condition holds precisely when $0<a\le 1/2$. In this parameter range, the Julia set and the anti‑Julia set are identical as point sets in the complex plane, despite arising as the closures of two generating sets that share no points except $\infty$. See Appendix~\ref{appendix_DML_infinity} for more detail.

\subsection{Asymptotic Approach Along Algebraic Curves}
The roots of the equations $A_a^{(n)} (z)=0$ and $D_a^{(n)}(z)=0$ define the trajectories of the finite periodic points and the pre-poles respectively, parameterized by the real deformation parameter $a$. By eliminating $a$, one can derive algebraic curves in the complex plane that dictate these trajectories. The analysis of these curves reveals a definitive pairing mechanism: a specific $n$-th periodic point $p_j(a)\in P_a^{(n)}$ and its corresponding $n$-th anti-periodic point $q_j(a)\in Q_a^{(n)}$ move along their respective algebraic curves and coincide in the limit $a \to 0$
\begin{align}
\lim_{a \to 0} p_j(a) = \lim_{a \to 0} q_j(a) \eqqcolon z_j^{\ast}.
\end{align}
At this intersection point $z_j^{\ast}$, both the numerator and the denominator of the periodic equation vanish simultaneously as will be shown in Appendix~\ref{appendix_cancellation}:
\begin{align}
A_0^{(n)}(z_j^{\ast}) = 0, \quad D_0^{(n)}(z_j^{\ast}) = 0.
\end{align}

Consequently, as $a\to 0$, the finite periodic points and pre-poles do not move independently. Instead, each finite periodic point and its corresponding pre-poles approach each other along their respective algebraic curves. They exhibit a strong correlation, as indicated by the following recurrence relations:
\begin{widetext}
\begin{align}
D_a^{(n+1)}&=D_a^{(n)}\Big(\big(1+(1-a)\mu z\big)D_a^{(n)}+(1-a)\mu zX(z)A_a^{(n)}\Big),\\
A_a^{(n+1)}&=-\Big(D_a^{(n)}\Big)^2+\big(\mu-\mu(1+a)z\big)A_a^{(n)}D_a^{(n)}-a\mu zX(z)\Big(A_a^{(n)}\Big)^2.
\end{align}
\end{widetext}
Since the points are continuously linked through these intersecting curves, it is reasonable to expect that an asymptotic one‑to‑one correspondence will arise.

%%
%% sec. 4
%%
\section{The ``Big Crunch'' at the Integrable Limit}\label{sec:bigcrunch}
\subsection{Algebraic Cancellation of the Dual Structures}\label{sec:cancellation}
As established in the previous section, the elements of the disjoint Cantor sets move along the algebraic curves and coincide at the intersection points, $z_j^{\ast}$ in the limit $a\to 0$. Algebraically, this implies that the polynomials $A_a^{(n)} (z)=0$ and $D_a^{(n)}(z)=0$ acquire a massive number of common factors at the integrable limit, as shown in detail in Appendix~\ref{appendix_cancellation}. To formulate this collision, we numerically factorize the numerator and the denominator of the periodic equation for some fixed values of $a$ and $\mu$, by separating the colliding pairs from the rest of the components:
\begin{align}
A_a^{(n)}(z) &= R_a^{(n)}(z) \prod_{j=1}^{2^n-2} (z - p_j(a)), \\
D_a^{(n)}(z) &= S_a^{(n)}(z) \prod_{j=1}^{2^n-2} (z - q_j(a)).
\end{align} Here, $p_j(a)\in P_a^{(n)}$ and $q_j(a)\in Q_a^{(n)}$ denote the pairs of the finite periodic points and the pre-poles, while $R_a^{(n)}(z)$ and $S_a^{(n)}(z)$ are polynomials representing the residual roots that do not participate in this specific pairing.

In the non-integrable regime, $a\neq 0$, the polynomials $A_a^{(n)} (z)$ and $D_a^{(n)}(z)$ are relatively prime. However, as $a\to 0$, the roots $p_j(a)$ and $q_j(a)$ converge to the identical simultaneous zeros, $z_j^{\ast}$. Consequently, the numerator and the denominator generate a common polynomial factor of degree $2^{n-1}-1$. Taking the limit of the rational periodic equation, this common factor is algebraically canceled:
\begin{align}
\lim_{a \to 0} \frac{A_a^{(n)}(z)}{D_a^{(n)}(z)} = \frac{R_0^{(n)}(z) \prod_{j=1}^{2^{n-1}-1} (z - z_j^{\ast})}{S_0^{(n)}(z) \prod_{j=1}^{2^{n-1}-1} (z - z_j^{\ast})} = \frac{R_0^{(n)}(z)}{S_0^{(n)}(z)}.
\end{align}
This algebraic cancellation constitutes the definitive mechanism of the collision and annihilation of the dual structures. The infinite number of the repelling periodic points dense in the Julia set, along with their counterpart pre-poles, logically vanish from the phase space through this algebraic reduction of the degree of the rational function.

The disappearance of the Cantor sets via this massive cancellation guarantees the transition from the chaotic state to the integrable state. 
The remaining polynomials, $R_0^{(n)}(z)$ and $S_0^{(n)}(z)$, which survive the annihilation, determine the regular periodic structure and the divergent
structure of the completely integrable system at $a=0$. In fact, we find
\begin{align}
R_0^{(n)}(z)=-\frac{1-\mu^n}{1-\mu},\qquad
S_0^{(n)}(z)=1+\mu\frac{1-\mu^n}{1-\mu}z,
\end{align}
so that $f_0^{(n)}(z)$ of \eqref{fantof0n} is recovered, as shown in Appendix~\ref{appendix_cancellation}.

\subsection{Remnants of the Annihilation: IDPs as Singular Loci}
The algebraic cancellation, established in the preceding subsection, dictates that the collision of the dual structures leaves algebraic remnants. As defined in Sec.~\ref{sec:antijulia}, the set of finite periodic points, $P_a^{(n)}$, and the set of pre-poles, $Q_a^{(n)}$, satisfy $P_a^{(n)} \cap Q_a^{(n)} = \emptyset$ for $a \neq 0$. At $a=0$, the convergence of the roots to the points $z_j^{\ast}$ translates to a non-empty intersection of these two sets:
\begin{align}
P_0^{(n)} \cap Q_0^{(n)} = \{ z_1^{\ast}, z_2^{\ast}, \dots, z_{n-1}^{\ast} \}.
\end{align}

At these intersection points, both the numerator and the denominator of the periodic equation vanish simultaneously:
\begin{align}
A_0^{(n)}(z_j^{\ast}) = 0, \quad D_0^{(n)}(z_j^{\ast}) = 0.
\end{align}
By definition, these simultaneous roots constitute the IDPs of the rational map. While the infinite number of roots dense in the Julia set and the anti-Julia set vanish from the phase space through the algebraic cancellation, the intersection points $z_j^{\ast}$ remain at $a=0$ as algebraic singularities.

For the DLM, $f_a(z)$, the coordinates of these remnants can be derived by analyzing the factorized polynomials $A_0^{(n)}(z)$ and $D_0^{(n)}(z)$ as shown in Appendix~\ref{appendix_cancellation} (see also~\cite{Saitoh1996} for an earlier preliminary mention). The algebraic reduction reveals that the simultaneous roots form a finite discrete set, designated as the set of IDPs. For a given period $n$, the set of IDPs consists of $n-1$ points given by:
\begin{align}
z_k^{\ast} = \frac{-1}{\sum_{m=1}^{k} \mu^m}, \quad (k = 1, 2, \dots, n-1).
\end{align}
This algebraic result demonstrates that the disjoint Cantor sets of the non-integrable regime degenerate into this set of IDPs at $a=0$, with the degeneracy $2^{n-k-1}$ at $z_k^*$ for period $n$.

The fact that these IDPs remain as the remnants of the annihilation provides the direct algebraic origin of the integrable structures. As formulated in Sec.~\ref{sec:duality}, when the map restarts from these singular loci, the process of singularity confinement iteratively generates the continuous IVPPs. Thus, the simultaneous zeros $z_{k}^{\ast}$ serve as the algebraic source from which the regular periodic structures of the completely integrable system emerge, completing the description of the integrable--nonintegrable transition.

\subsection{Numerical Visualizations of the Transition}
The finite-depth backward-orbit approximations of the pre-Julia and pre-anti-Julia sets for the DLM are executed by iterating the identical inverse map, $z_{i+1}=f_{a}^{(-1)}(z_{i})$, generating $2^{i}$ points at the $i$-th iteration. Despite sharing this generating rule, the two sets originate from different initial values, $z_{0}$. The pre-Julia set starts from the finite repelling fixed point $z_{0}^{\mathrm{(J)}}=1-1/\mu$ (for $\mu>1$), which is independent of $a$. Conversely, the pre-anti-Julia set originates from the pole of the rational map at infinity, $z_{0}^{\mathrm{(A)}} =\infty$, yielding $z_{1}^{\mathrm{(A)}} = -1/(\mu(1-a))$ after the first backward iteration.

At each step, the inverse map $z_{i+1}$ is evaluated by finding the roots of the quadratic equation:
\begin{align}
a z_{i+1}^2 - \left( 1 - (1 - a) z_i \right) z_{i+1} + \frac{z_i}{\mu} = 0. \label{eq:inverse_map}
\end{align}
Evaluating the smaller root via the standard quadratic formula induces catastrophic cancellation as $a \to 0$, causing an extreme loss of significance in IEEE 754 double-precision floating-point arithmetic. To track the chaotic boundaries down to arbitrary scales of $a$ without artificially destroying the fractal structures, we must reformulate the roots into a numerically stable form~\cite{Goldberg1991}. By isolating the core complex variable components into a residual rational term, $\mathcal{R}_i = 1 - (1 - a) z_i$, and defining a sign-enforcing intermediate algebraic variable
\begin{align}
T_i = \frac{1}{2} \left( \mathcal{R}_{i} + \text{sgn}\left(\Re(\mathcal{R}_{i})\right) \sqrt{\mathcal{R}_{i}^{2} - \frac{4 a z_i}{\mu}} \right),
\end{align}
the recurrence relation for the two bifurcated branches is stabilized as:
\begin{align}
z_{i+1} =
\begin{cases}
\dfrac{T_i}{a} & (\text{plus branch}), \\[3ex]
\dfrac{z_i}{\mu \cdot T_i} & (\text{minus branch}).
\end{cases}
\end{align}
This algebraic substitution completely removes the small parameter $a$ from the divisor of the small root calculation, successfully suppressing numerical errors to the level of machine epsilon ($\approx 10^{-16}$).

Because the inverse map acts as a contraction mapping, the macroscopic difference between the initial values, $z_{0}^{\mathrm{(J)}}$ and $z_{0}^{\mathrm{(A)}}$ rapidly diminishes. However, for any non-integrable regime, $a \neq 0$, both the finite-depth approximates of the pre-Julia set and the pre-anti-Julia set never intersect, maintaining a strictly bijective, one-to-one correspondence separated by a microscopic spatial gap scaling as, $\mathcal{O}(a)$. Preserving this exact pairing is crucial for comparing the dual sets. Consequently, stochastic sampling methods, like the Randomized Inverse Iteration Method (RIIM), which exponentially undersample specific regions and destroy the synchronized bijection, are inadequate. Instead, we must adopt a deterministic full-tree expansion that exhaustively computes and pairs all $2^{i}$ branches at each iteration.

This strict pairing requirement induces a conflict between topological resolution and algorithmic complexity $\mathcal{O}(2^{i})$. To visually resolve the $\delta_a = \mathcal{O}(a)$ gap, the spacing between adjacent points within the generated sets ($\Lambda^i$, where $\Lambda = e^{-\lambda}$ is the average contraction rate strictly determined by the Lyapunov exponent $\lambda>0$) must be smaller than $\delta_a$. Pursuing the integrable limit down to the ultimate precision of the stabilized double-precision arithmetic ($\delta_a \sim a \sim \epsilon_{\text{machine}} = 2^{-53}$) requires $e^{-\lambda i} \le 2^{-53}$. To estimate the absolute minimum number of iterations required to reach this limit, we consider the most rapidly contracting scenario: a fully developed chaotic regime ($\lambda = \ln 2$, i.e., $\Lambda = 1/2$). Even in this best-case scenario, which yields the theoretical lower bound for the iteration depth, reaching the machine precision limit necessitates $i=53$. Consequently, preserving the exact bijective pairing down to this scale demands the computation of $2^{53} \approx 9 \times 10^{15}$ exact algebraic pairs. Because this minimal depth required to expose numerical vulnerabilities far exceeds the practical iteration depth bounded by combinatorial explosion, and since backward iterations inherently suppress chaotic error amplification, accumulated computational errors remain entirely negligible for any realistically executable iteration depth.

\begin{figure*}[p]
\centering
\includegraphics[width=0.32\textwidth]{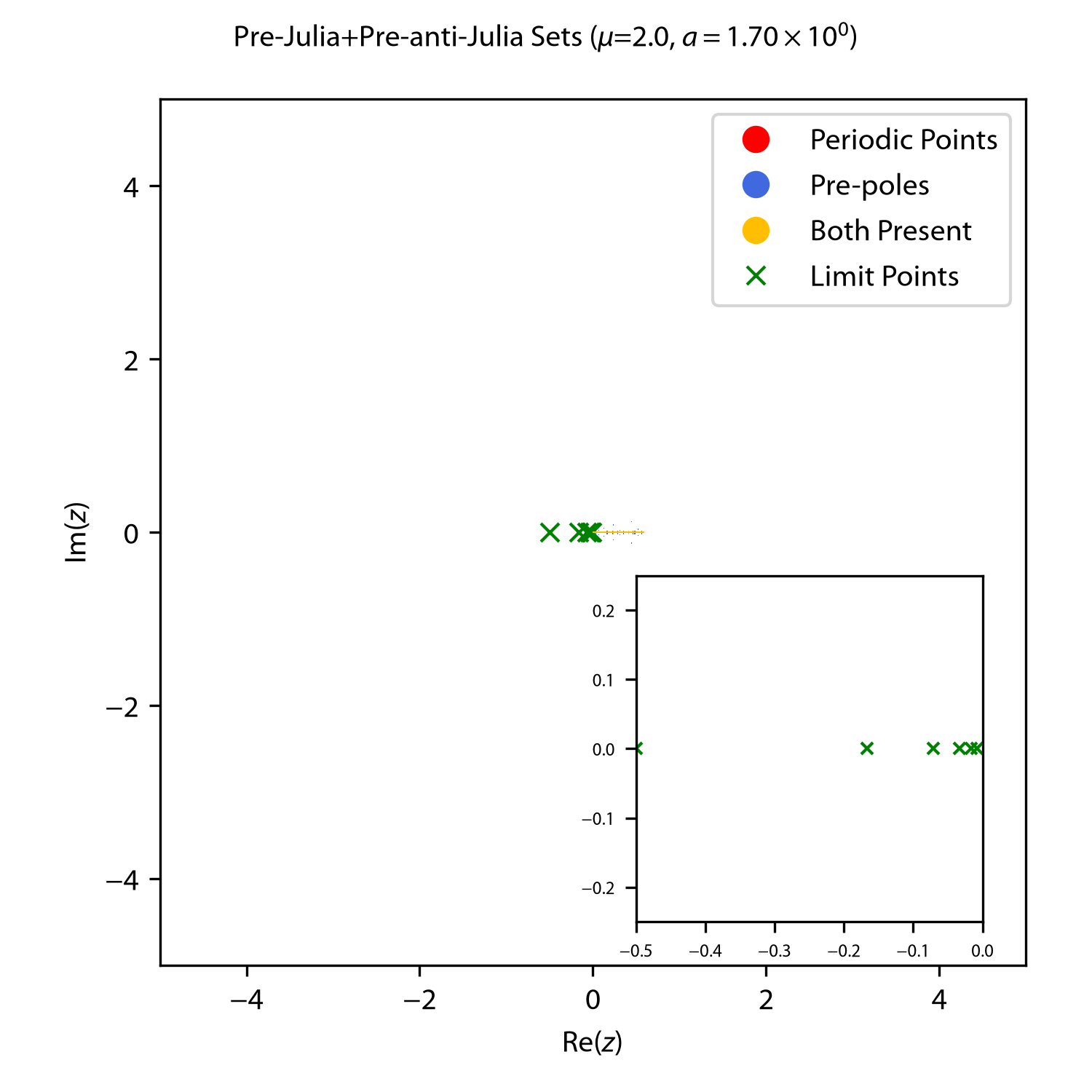}
\includegraphics[width=0.32\textwidth]{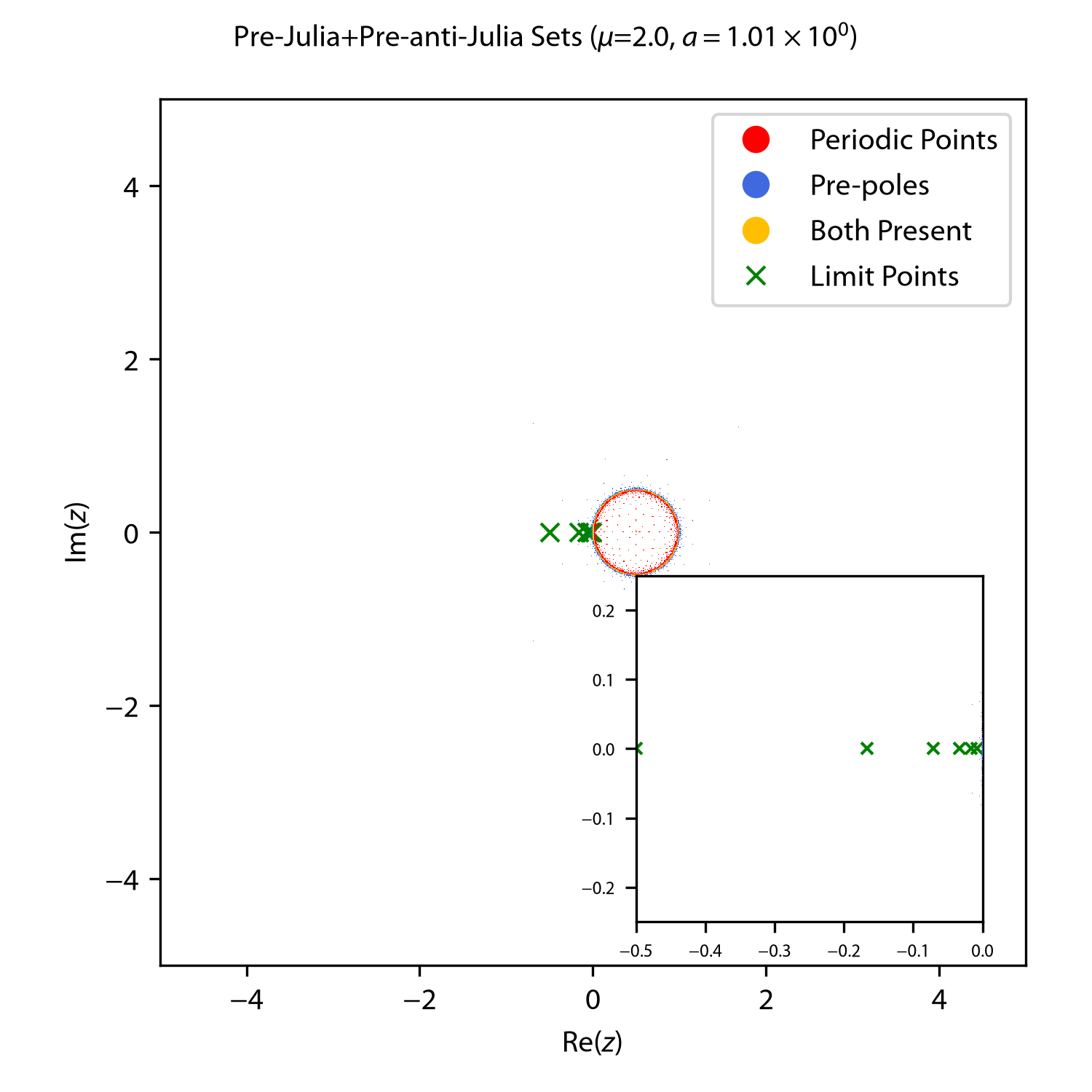}\\
\includegraphics[width=0.32\textwidth]{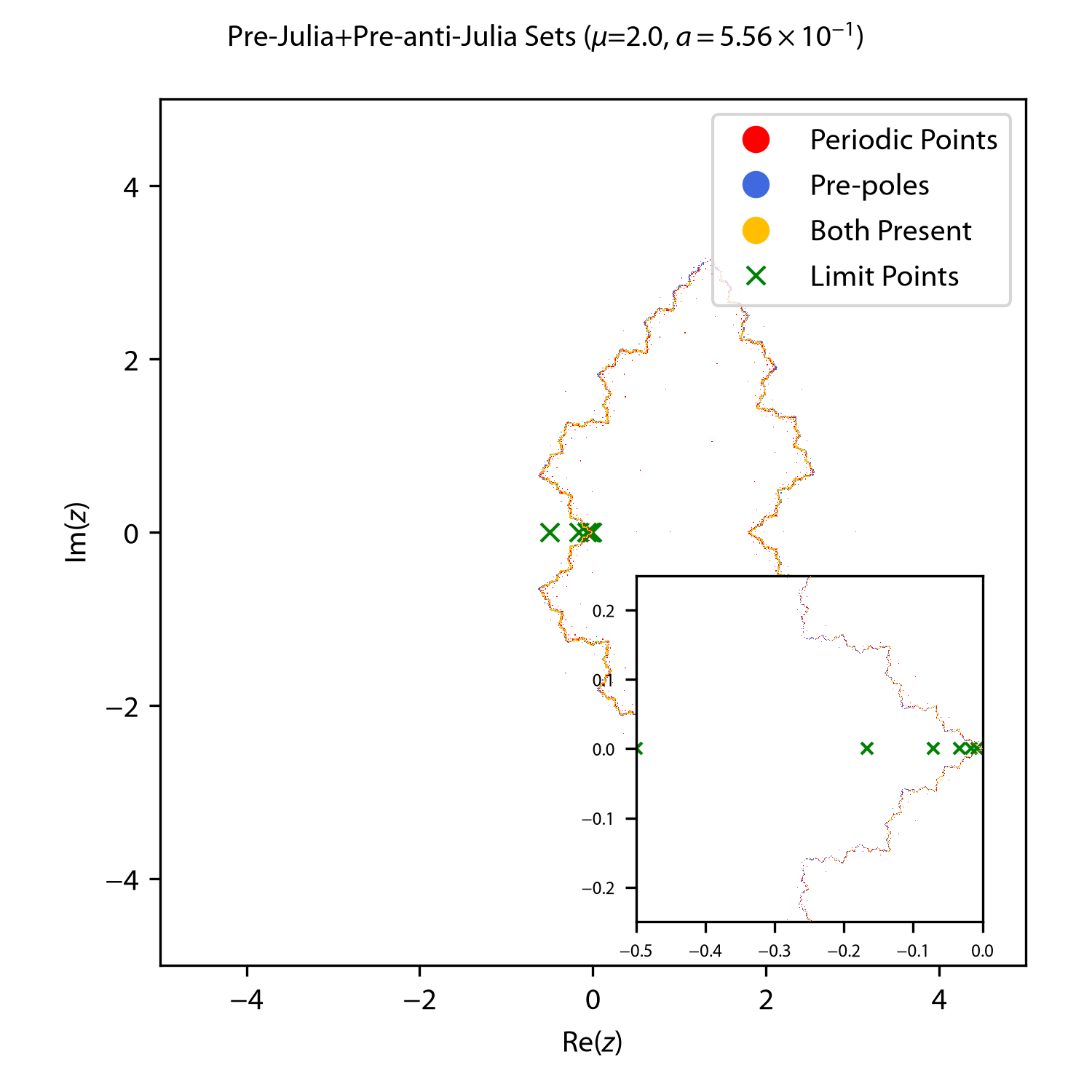}
\includegraphics[width=0.32\textwidth]{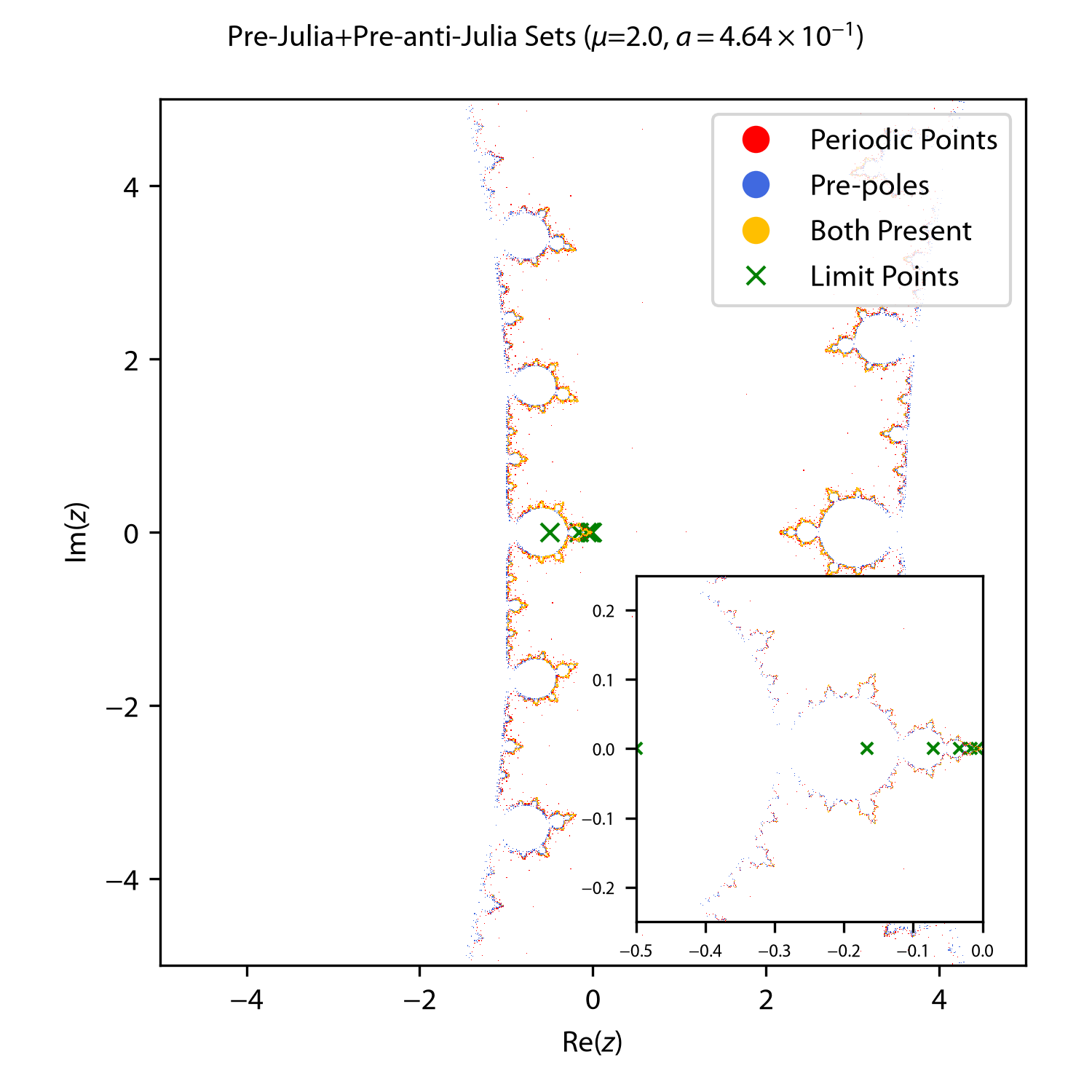}\\
\includegraphics[width=0.32\textwidth]{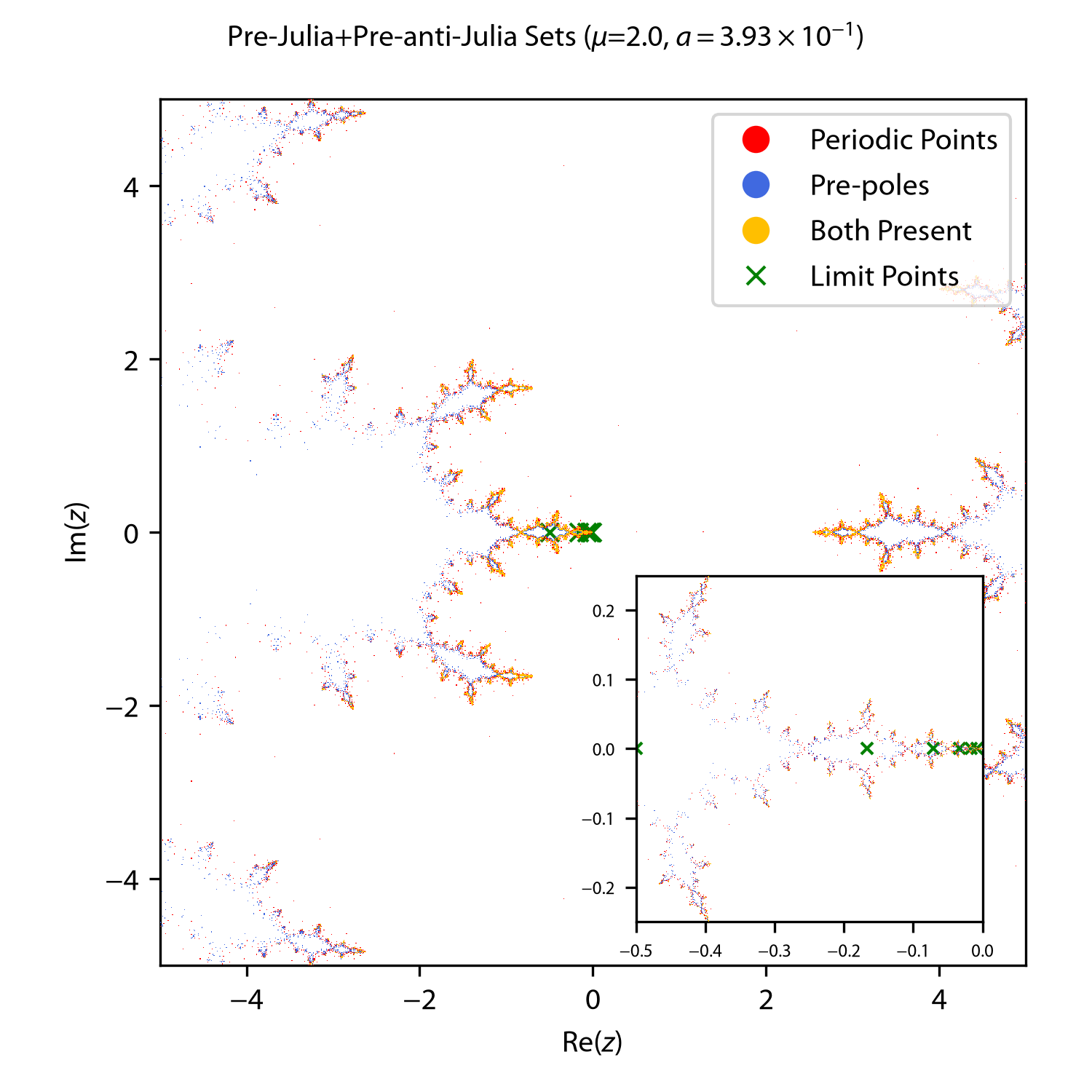}
\includegraphics[width=0.32\textwidth]{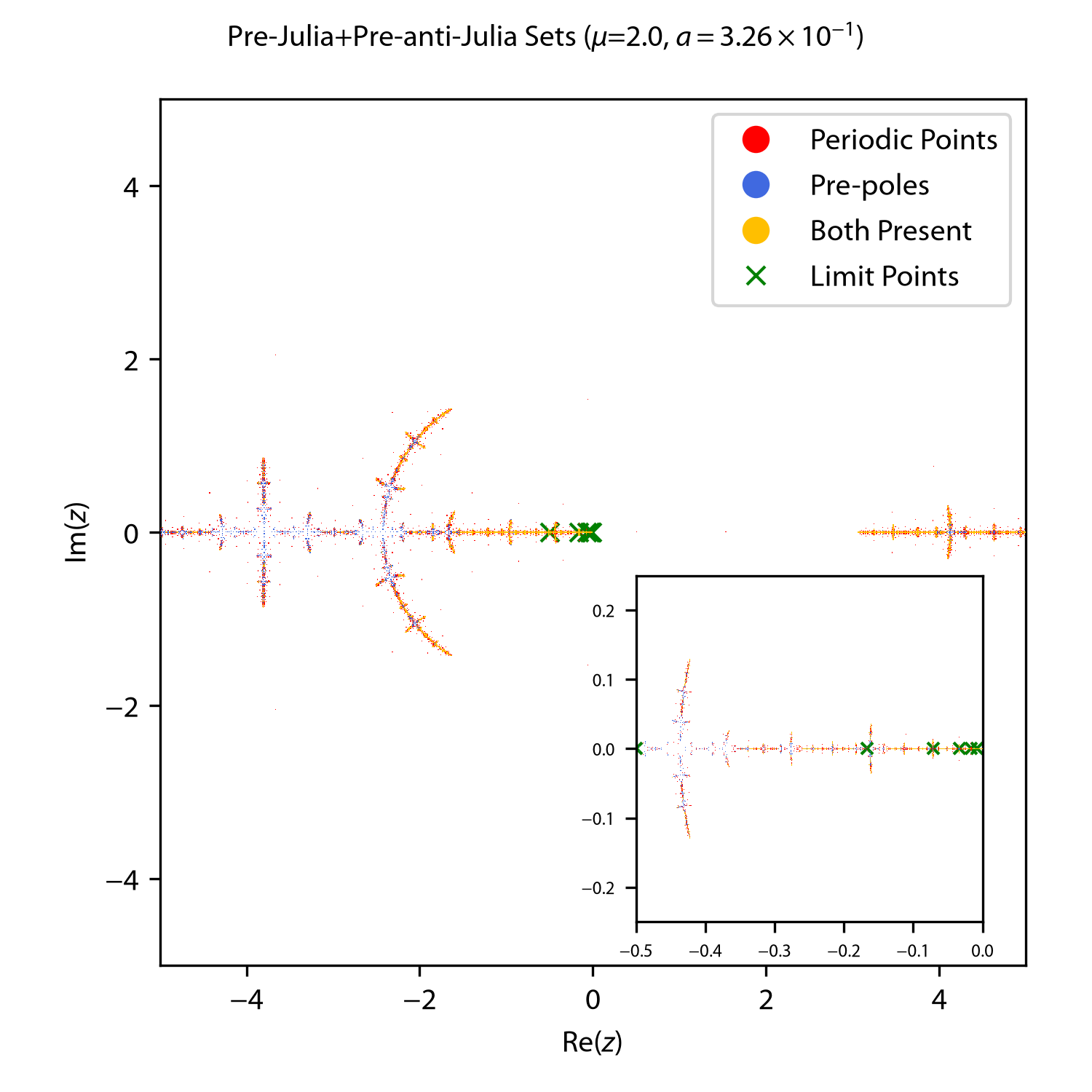}\\
\includegraphics[width=0.32\textwidth]{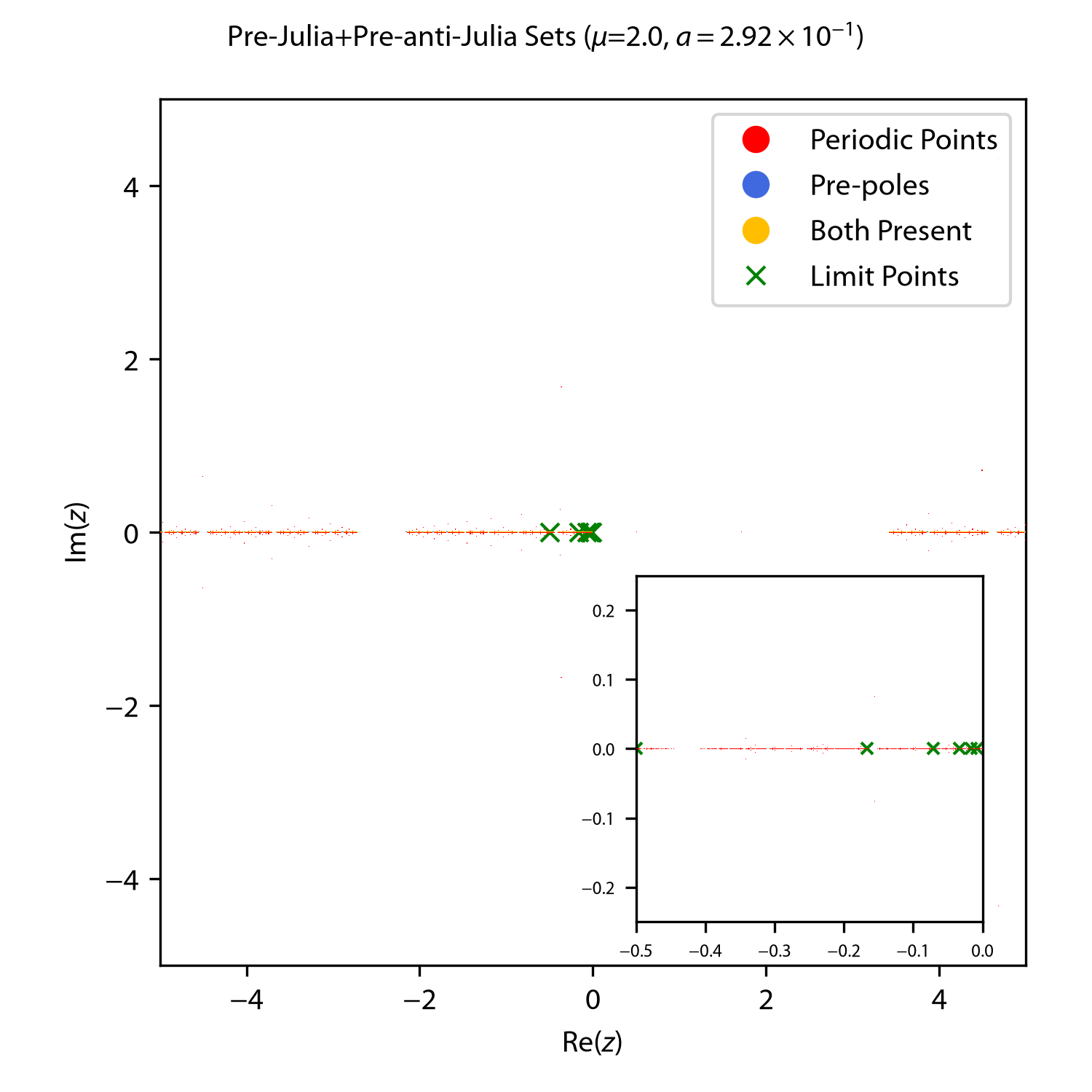}
\includegraphics[width=0.32\textwidth]{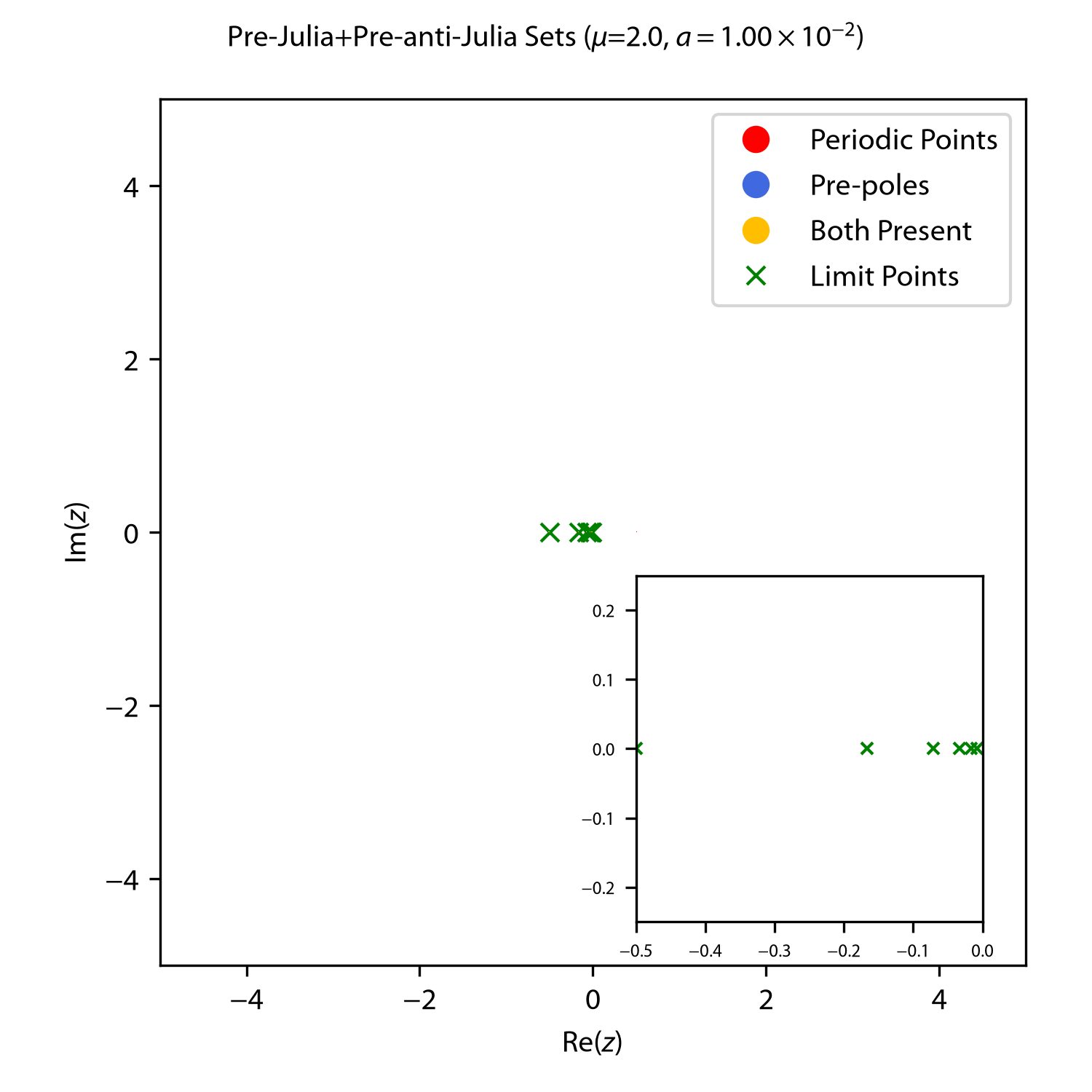}
\caption{Selected frames show the DLM transition under specific conditions, plotting the elements of the finite-depth approximates of the pre-Julia set and the pre-anti-Julia set  via deterministic full-tree expansion in 15 steps. ``Both Present'' indicates pixels where the periodic points and the pre-poles coexist, not where they coincide. As shown in the inset, zooming in reduces ``Both Present'' pixels. Although the two sets do not coincide for $a \neq 0$, visualizing this non-coincidence requires infinite resolution. The full animation sequence is available in the Zenodo repository~\cite{Iizawa2026_bigcrunch_zenodo}.}\label{fig:anti-julia}
\end{figure*}

Fig.~\ref{fig:anti-julia} presents numerical plots of the finite-depth approximates of the pre-Julia set and the pre-anti-Julia set, illustrating the preceding discussion: for $a \neq 0$, the topological difference originating from the distinct initial values is strictly preserved as an $\mathcal{O}(a)$ gap. As $a \to 0$, this gap systematically collapses. At $a=0$, the dual structures generated from the infinite pole and the finite fixed point degenerate simultaneously into an identical set of IDPs. This confirms that the initial-value dependency is completely lost through deterministic algebraic cancellation at the integrable limit.

\begin{figure*}[ht]
\centering
\includegraphics[width=1.0\textwidth]{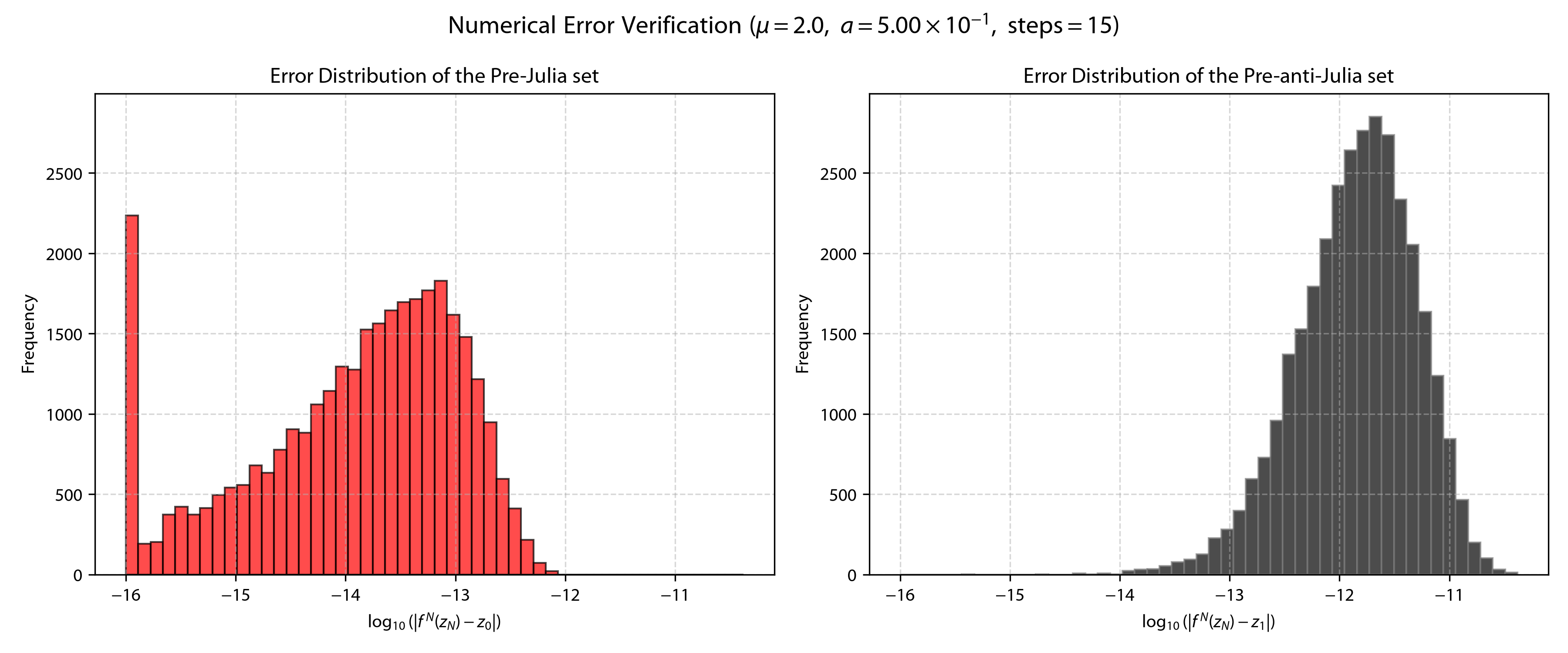}
\includegraphics[width=0.8\textwidth]{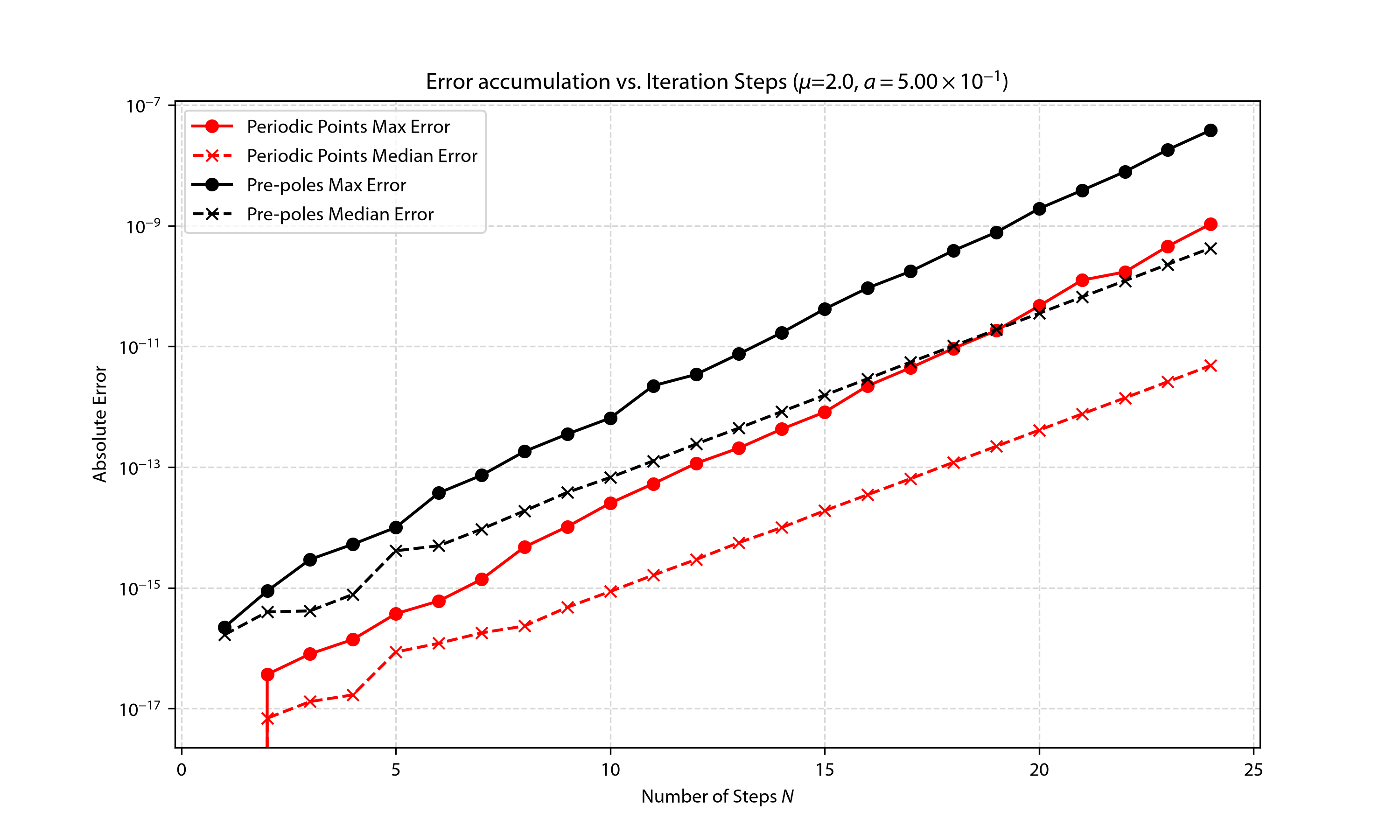}
\caption{Numerical error analysis for the finite-depth approximates of the pre-Julia set and the pre-anti-Julia set derivations. Forward mappings were applied for the same number of steps following inverse mappings to evaluate the difference; thus, this yields an estimate much more conservative than the actual numerical error.}\label{fig:error}
\end{figure*}

To quantitatively validate the numerical stability of the algorithm and assess the accumulated computational errors, we performed a round-trip error analysis, as illustrated in Fig.~\ref{fig:error}. Specifically, after generating the set of points $\{z_N\}$ through $N$ steps of iterations applying inverse map, we apply the forward mapping $f_a(z)$ for $N$ consecutive steps to reconstruct the initial values, yielding $\{\hat{z}_0\}$. The statistics of the discrepancy $|\hat{z}_0 - z_0|$ serves as our error metric. Because the forward map of DLM $f_a(z)$ is chaotic and acts as an expansive mapping that exponentially amplifies machine-precision fluctuations, this round-trip evaluation yields an estimate much more conservative than the actual error present in the generated fractal sets.

The upper panel of Fig.~\ref{fig:error} displays the distribution of these reconstructed errors for both sets. For the pre-Julia set, the vast majority of errors remain near the machine precision limit ($\sim 10^{-16}$), with the remainder smoothly distributed up to $10^{-12}$. In contrast, the error distribution for the pre-anti-Julia set peaks near $10^{-12}$, spanning roughly between $10^{-14}$ and $10^{-11}$. This upper limit of $\sim 10^{-12}$ at $N=15$ corresponds directly to the maximum local expansion rate of the forward mapping, which amplifies machine-precision fluctuations ($\sim 10^{-16}$) by a factor up to $2^{15} \approx 3.3 \times 10^4$. Additional numerical checks confirm that varying $\mu$ by a factor of a few or altering $a$ across several orders of magnitude yields no error growth that would visually distort the fractal structures.

The lower panel of Fig.~\ref{fig:error} shows the median and maximum errors as functions of the iteration step $N$. Although the accumulated error increases exponentially with the step number, practical computations on standard personal computers are constrained to $N \le 24$ due to memory limitations, as generating $2^{24}$ elements for each set ($\approx 3.35 \times 10^7$ complex numbers in total) approaches the handling capacity of standard visualization pipelines and hardware. Consequently, within any executable iteration depth on conventional PC setups, the exponential error growth remains entirely insufficient to introduce visual artifacts or distort the fractal boundaries. We note that localized computational anomalies may occasionally emerge depending on the specific choices of $\mu$ and $a$; however, these isolated deviations do not alter the overall behavior.

%\afterpage{
%\clearpage
\begin{figure*}[p]
\centering
\includegraphics[width=1.0\textwidth]{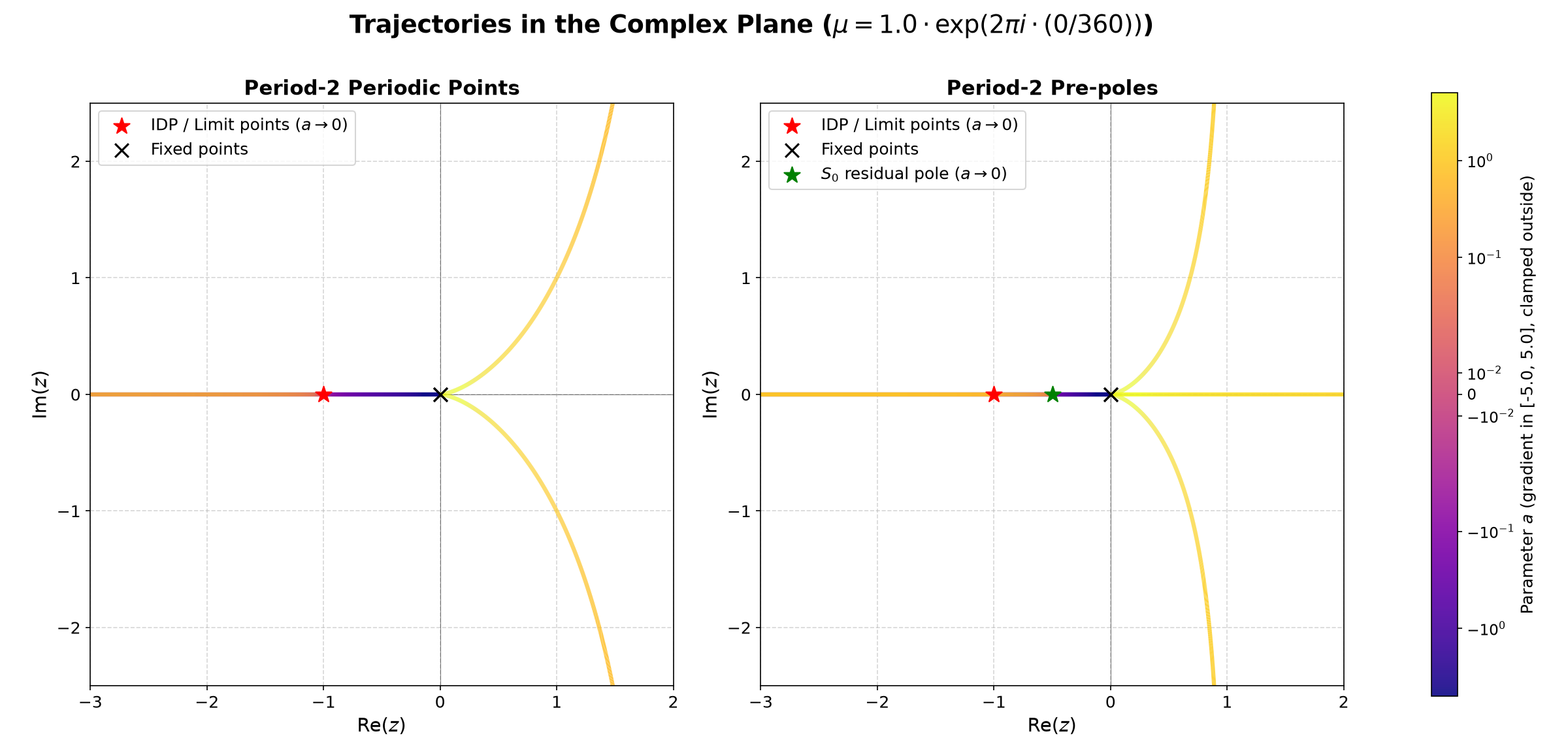}
\includegraphics[width=1.0\textwidth]{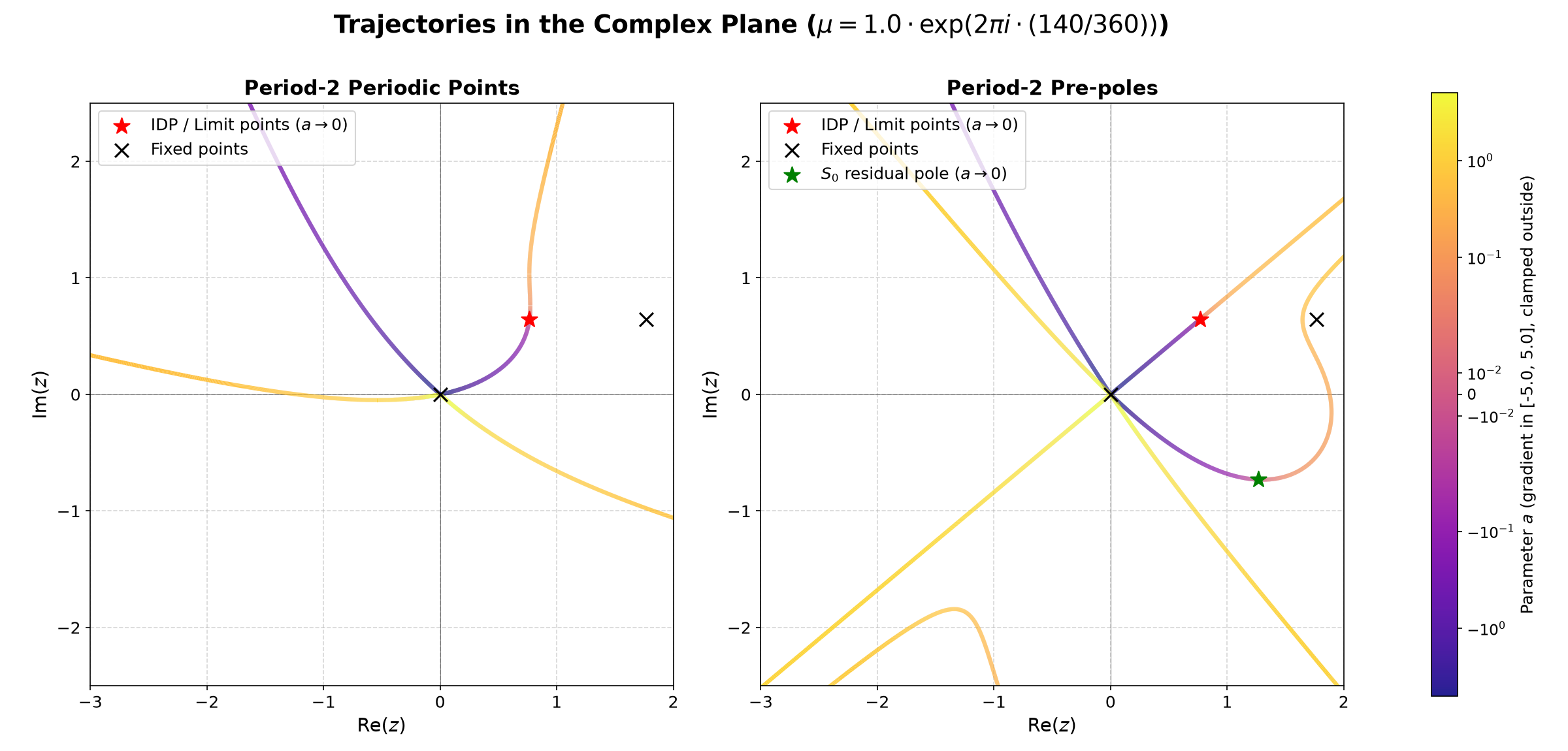}
\caption{Algebraic curves of the finite-depth approximates of periodic points the (left panels) and pre-poles (right panels) for period 2, parameterized by $a$ (color-coded along the trajectories). The upper panels correspond to $\theta = 0^{\circ}$ (real $\mu$), and the lower panels to $\theta = 140^{\circ}$. The green star is a zero point derived from $S_{0}^{(n)}$. Consequently, since the curve passing through the green star is derived from an uncancelled term of $D_a^{n}$, it is not a pre-pole curve. For any $a \neq 0$, curves sharing the same color never intersect; however, as $a \to 0$, all branches systematically converge and collide at the origin ($a = 0$). Varying the phase $\theta$ induces complex topological reconnections among the curves while strictly preserving the non-intersection rule for constant $a$. Note that these algebraic curves converge to the IDP as $a \to 0$, but since the periodic points (or pre-poles) have already annihilated on the IDP, the curves are not defined at the IDP itself (i.e., they are open curves). Similarly, parts of these algebraic curves converge to the origin (a fixed point) as $a \to \pm \infty$, but the origin itself is not part of the curves. Full animation: Zenodo repository~\cite{Iizawa2026_bigcrunch_zenodo}.}\label{fig:anti-periodic1}
\end{figure*}
%\clearpage
\begin{figure*}[p]
\centering
\includegraphics[width=1.0\textwidth]{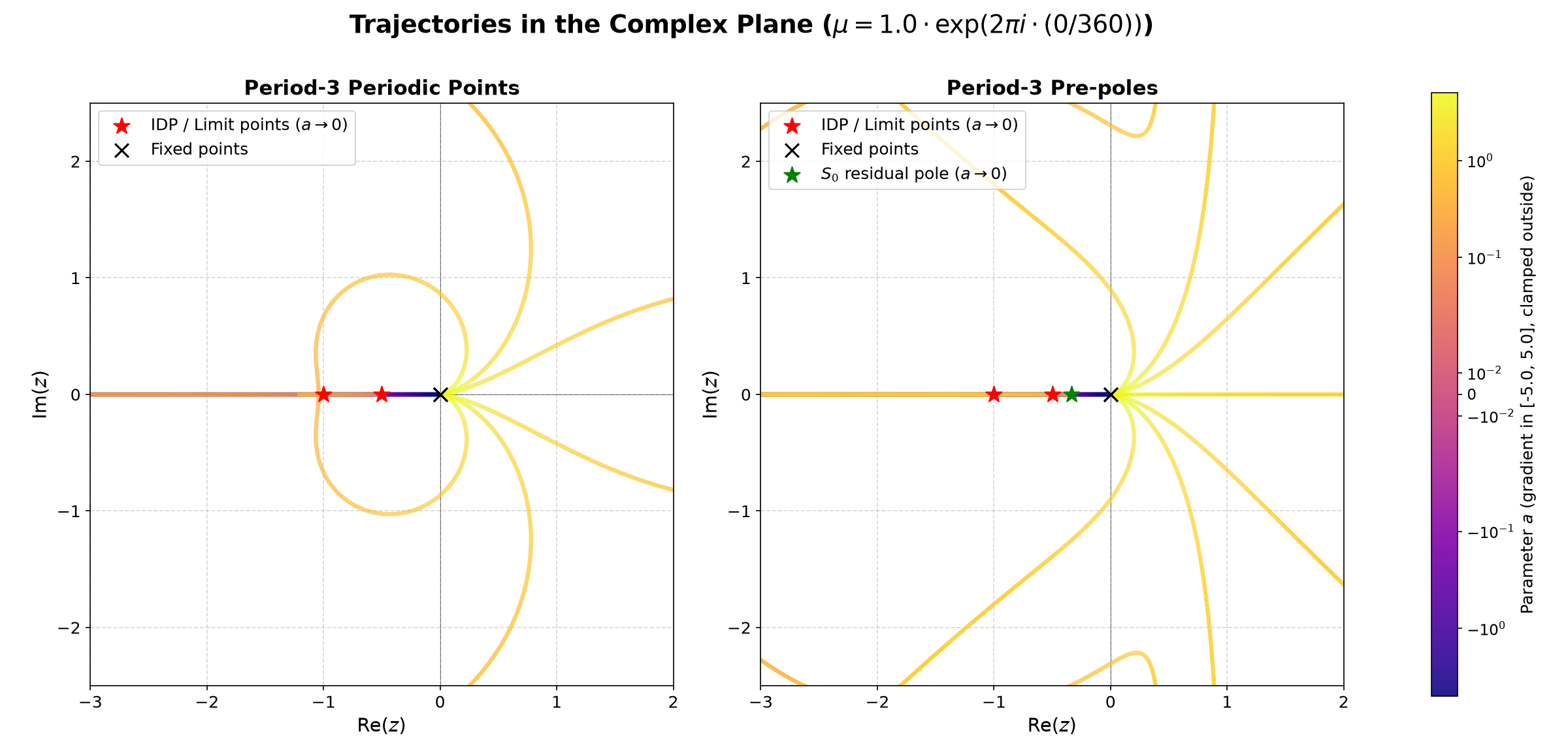}
\includegraphics[width=1.0\textwidth]{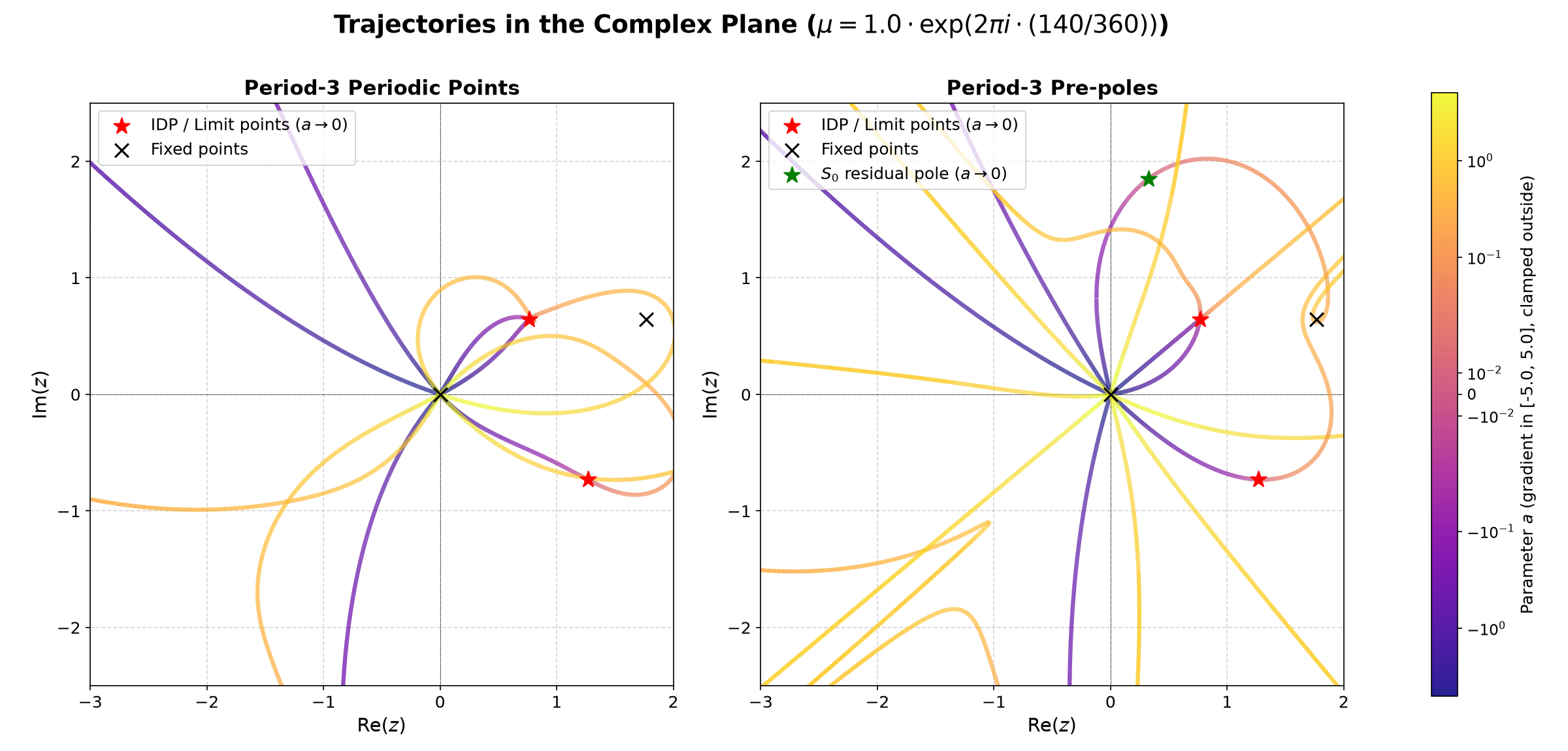}
\caption{Algebraic curves of the finite-depth approximates of  periodic points (left panels) and pre-poles (right panels)  for period 3.}\label{fig:anti-periodic2}
\end{figure*}
%\clearpage
\begin{figure*}[p]
\centering
\includegraphics[width=1.0\textwidth]{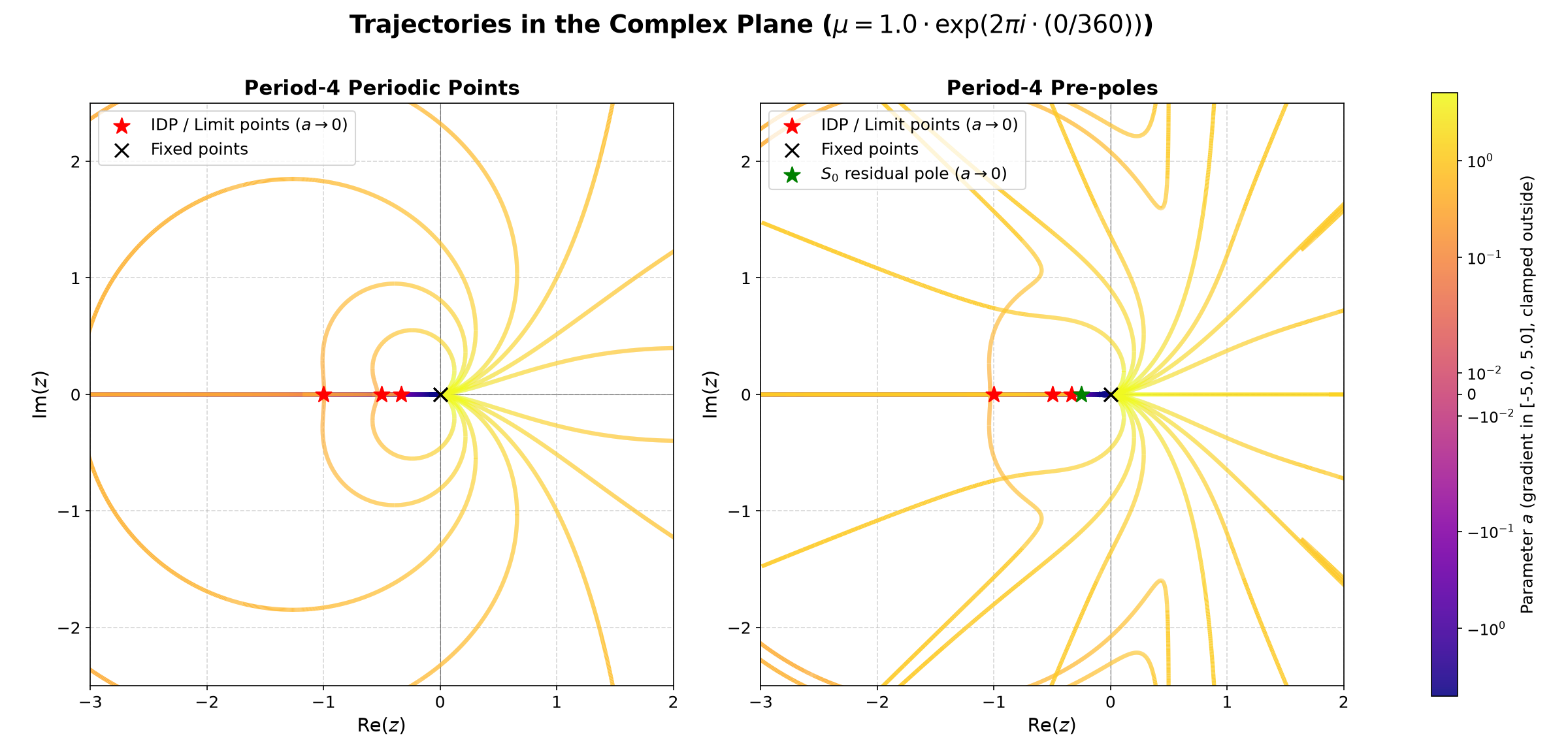}
\includegraphics[width=1.0\textwidth]{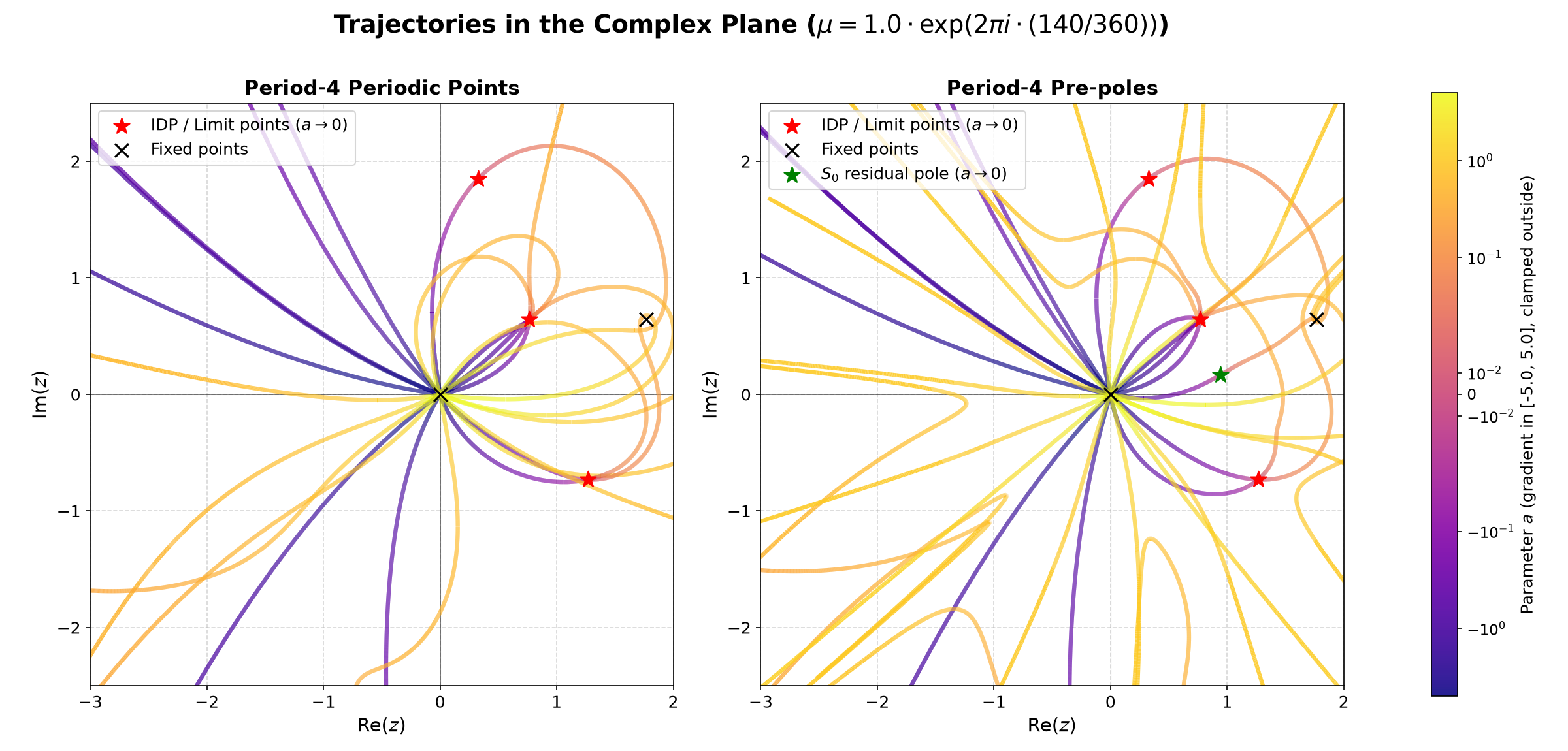}
\caption{Algebraic curves of the finite-depth approximates of periodic points (left panels) and pre-poles (right panels)  for period 4. (Note: Numerical tracing is restricted to periods up to 4; for period 5 and higher, the excessively high polynomial degree induces severe numerical instability in double-precision arithmetic, preventing accurate rendering.)}\label{fig:anti-periodic3}
\end{figure*}
\clearpage
%}

To reveal the essential geometry driving this transition, we trace the algebraic curves of periodic points and pre-poles generating the pre-Julia and pre-anti-Julia sets, respectively parameterized by $a$, as shown in Fig.~\ref{fig:anti-periodic1}--\ref{fig:anti-periodic3}. The value of $a$ is represented by color along these trajectories. For any non-integrable state ($a \neq 0$), curves never self-intersect at points sharing the same color. This strict prohibition of same-color intersections is fundamental: it demonstrates that despite their intricate complex trajectories, the pre-Julia set and the pre-anti-Julia set remain topologically disjoint for any $a \neq 0$.  As $a \to 0$, however, this topological distinction completely vanishes. The algebraic branches irresistibly collapse toward one another, culminating in a total collision at $a = 0$---a phenomenon we term the \textit{Big Crunch}. Rotating the phase $\theta$ drives non-trivial transformations across the complex plane, yet the core mechanism remains invariant: no matter how $\theta$ is varied, same-$a$ intersections are strictly forbidden, confirming that the degeneration at $a = 0$ is an intrinsic property of the integrable limit, invariant under phase variations.

%%
%% sec. 5
%%
\section{Conclusion: A Unified Picture of the Integrable--Nonintegrable Transition}\label{sec:conclusion}
By introducing the anti-Julia set, defined as the closure of the divergent poles, we revealed that the non-integrable phase space accommodates two disjoint generating sets, originating from the numerator and the denominator of the periodicity conditions, whose closures form coincident Cantor sets for small $a$.  As the system approaches the completely integrable state, these dual discrete structures converge along algebraic curves and collide at the IDPs.

The disappearance of chaos at the integrable limit is identified as the massive algebraic cancellation between these colliding sets. The infinite number of repelling periodic points and their divergent counterparts annihilate each other, leaving only the discrete IDPs as the singular loci of the map. These simultaneous zeros serve as the algebraic source of the regular integrable structures. This deterministic collision and annihilation of the dual fractal sets provides a unified algebraic picture of the transition between chaos and integrability.

As noted in Sec.~\ref{subsec:necessity_dual}, our anti‑Julia set is defined as the closure of the set of pre‑poles, and hence its behaviour is highly sensitive to the dynamics of the map at $\infty$. As explained in Appendix~\ref{appendix_DML_infinity}, when 
$0<a<1/2$ the anti‑Julia set coincides with the Julia set as a set, even though their defining numerator and denominator differ. Therefore, the anti‑Julia set is not identical to the Julia set in a dynamical sense, even though they coincide as point sets. Moreover, when $1/2<a<1$ the pre-anti-Julia set is not on the Julia set, as it is shown in Appendix B. These facts require that they be referred to by their respective names. Nevertheless, at $a=0$, corresponding to the integrable limit, the two dual sets collapse and degenerate into the same set of singular loci (IDPs).

Beyond these topological insights, the systematic error analysis established in this study provides a crucial quantitative benchmark. While early numerical estimates of the fractal dimensions were constrained by the computational limits of their time, our error evaluation lays a solid foundation for future high-precision dimensional analyses, offering a clear path toward fully unravelling the intricate scaling behavior near the integrable limit.

Extending our analysis to higher-dimensional systems represents a natural and promising direction for future research. In conventional near-integrable systems focusing on local perturbations, the presence of Arnold diffusion is well known, where trajectories undergo long-term global drift despite the persistence of a majority of KAM tori. It remains an open and intriguing question how our global approach using deformation maps behaves in the integrable limit for higher dimensions, and whether it can capture or encounter structural breakdown due to multi-dimensional resonances.

\begin{acknowledgments}
We are deeply indebted to Dr. Shuichi Otake for his extraordinary contributions to this project. We acknowledge support by the Open Access Publication Funds of the Technische Universit\"{a}t Braunschweig.
\end{acknowledgments}

S.S. conceptualized the work; S.S., M.I., and A.S. developed the methodology; M.I. wrote the original draft; M.I., S.S., Y.N., Y.M., H.H., and A.S. reviewed and edited the manuscript. All authors read and approved the final manuscript.

The authors declare no competing financial interest.

\section*{Data Availability}
Supplementary videos for further interest are available in Zenodo~\cite{Iizawa2026_bigcrunch_zenodo}.

\appendix
\section{Proof of the Cancellation in the Integrable Limit}\label{appendix_cancellation}
We show here how the exact cancellation of all factors in both the numerator and denominator of the periodicity equation occurs for any period $n$.

From the definition of $N_a^{(n)}(z)$ and $D_a^{(n)}(z)$ in \eqref{fantof0n}, we obtain the following recursive relations:
\begin{align}
N_a^{(n+1)}(z)&=N_a^{(n)}(z)\big(D_a^{(n)}(z)-a\mu^n zN_a^{(n)}(z)\big),\label{Na}\\
D_a^{(n+1)}(z)&=D_a^{(n)}(z)\big(D_a^{(n)}(z)+(1-a)\mu^{n+1}zN_a^{(n)}(z)\big),\label{Da}
%A_a^{(n)}(z)&=(\mu^nN_a^{(n)}(z)-D_a^{(n)}(z))/X(z)\label{Aa}
\end{align}
with initial conditions
\begin{align*}
N_a^{(1)}(z)=1-az,\qquad  D_a^{(1)}(z)=1+\mu(1-a)z.
\end{align*}
We now investigate the limit as $a\to 0$. In this limit the reccurence relations \eqref{Na}, \eqref{Da} reduce to 
\begin{align}
N_0^{(n+1)}(z)&=N_0^{(n)}(z)D_a^{(n)}(z),\label{N0}\\
D_0^{(n+1)}(z)&=D_0^{(n)}(z)\big(D_0^{(n)}(z)+\mu^{n+1}zN_0^{(n)}(z)\big),
\label{D0}
\end{align}
with initial conditions
\begin{align}
N_0^{(1)}(z)=1,\qquad D_0^{(1)}(z)=1+\mu z.
\label{N0,D0}
\end{align}

\begin{lemma}
\begin{align}
\frac{D_0^{(n)}(z)}{N_0^{(n)}(z)}=1+\left(\sum_{k=1}^n\mu^k\right) z \eqqcolon K_{n}(z)
\label{D/N=K}
\end{align}
\end{lemma}

\begin{proof}
For $n=1$, we have
\begin{align*}
D_0^{(1)}(z)/N_0^{(1)}(z)=1+\mu z=K_1(z),
\end{align*}
so the identity \eqref{D/N=K} holds.
Assume \eqref{D/N=K} is valid for some $n$. 
Then
\begin{align*}
\frac{D_0^{(n+1)}(z)}{N_0^{(n+1)}(z)}&=
\frac{D_0^{(n)}(z)\big(D_0^{(n)}(z)+\mu^{n+1}zN_0^{(n)}(z)\big)}{N_0^{(n)}(z)D_0^{(n)}(z)}\\
&=\frac{D_0^{(n)}(z)}{N_0^{(n)}(z)}+\mu^{n+1}z=K_{n+1}(z).
\end{align*}
This completes the proof.
\end{proof}
\quad

Now we must solve $N_0^{(n)}(z)$ and $D_0^{(n)}(z)$ satisfying \eqref{N0} and \eqref{D/N=K}, but it is not difficult to check that
\begin{align}
\begin{split}
N_0^{(n+1)}(z)&=\prod_{k=1}^n\big(K^{(k)}(z)\big)^{2^{n-k}},\\
D_0^{(n+1)}(z)&=K^{(n+1)}(z)\prod_{k=1}^n\big(K^{(k)}(z)\big)^{2^{n-k}}.
\end{split}
\label{N0D0}
\end{align}
satisfy them. Substituting these expressions \eqref{N0D0} into 
\begin{align}
A_0^{(n+1)}(z)=\Big(\mu^{n+1}N_0^{(n+1)}(z)-D_0^{(n+1)}(z)\Big)/X(z) \label{A0}
\end{align}
yields
\begin{align}
A_0^{(n+1)}(z)=-\frac{1-\mu^{n+1}}{1-\mu}\prod_{k=1}^n\big(K^{(k)}(z)\big)^{2^{n-k}}.
\label{A0n+1}
\end{align}
Thus, all factors in the ratio $A_0^{(n+1)}(z)/D_0^{(n+1)}(z)$ cancel exactly, except tor the remaining factor $K^{(n+1)}(z)$ in the denominater, as expected.

\quad

Since the degree of our rational map $f_a(z)$ is 2, the degree of the $n$-th periodicity condition  is $2^n$. From this, we can count the degrees of $A_a^{(n)}(z)$ and $D_a^{(n)}(z)$ in 
\eqref{A/D} as 
\begin{align}
\mathrm{deg}A_a^{(n)}(z)=2^n-2,\qquad
\mathrm{deg}D_a^{(n)}(z)=2^n-1.
\end{align}
On the other hand, because each function $K^{(k)}(z)$ is linear in $z$, we obtain
\begin{align}
\mathrm{deg} A_0^{(n)}(z)=2^{n-1}-1,\qquad
\mathrm{deg} D_0^{(n)}(z)=2^{n-1}.
\end{align}
The difference between them arises from the residual facror $K^{(n)}(z)$ in $D_0^{(n)}(z)$. We also observe that $2^{n-1}-1$ factors disappear  from both $A_a^{(n)}(z)$ and $D_a^{(n)}(z)$ in the limit $a\to 0$, corresponding to the behaviour at $z\to\infty$.

From \eqref{N0D0} and \eqref{A0n+1}, there are the same number of degeneracy at each point of  IDP. The multiplicity of the degeneracy at 
\begin{align}
z_k^*=-\frac{1}{\sum_{m=1}^k\mu^m}
\end{align}
is $2^{n-k-1}$.

The same number of roots of $A_a^{(n)}(z)$ and $D_a^{(n)}(z)$ converge to each $z_k^*$ as $a$ approaches 0, and their motions are correlated. From the expressions of their derivatives with respect to $a$
\begin{widetext}
\begin{align}
\left.\frac{dD_a^{(n+1)}}{da}\right|_{a=0}&=\frac{D_0^{(n)}}{K^{(n)}}\left[\Big((X+2\mu)K^{(n)}+X\Big)\frac{dD_a^{(n)}}{da}+\mu zXK^{(n)}\frac{dA_a^{(n)}}{da} -\mu^{n+1}z
D_0^{(n)}\right],\\
\left.\frac{dA_a^{(n+1)}}{da}\right|_{a=0}&=\frac{D_0^{(n)}}{K^{(n)}}\left[-\Big(K^{(n)}+\frac{1-\mu^{n+1}}{1-\mu}\Big)\frac{dD_a^{(n)}}{da}+(1-X)K^{(n)}\frac{dA_a^{(n)}}{da} +\mu^n(K^{(n)}-1)\frac{D_0^{(n)}}{K^{(n)}}\right],
\end{align}
\end{widetext}
they vanish simultaneously at $a=0$. Consequently a half of the roots reach the IDP at the same moment and stop there, and annihilate simultaneouslly.

\section{Behaviour of the DML at Infinity}\label{appendix_DML_infinity}
Consider a rational map of $d\ge 2$ 
\begin{align}
f: \hat{\mathbb{C}} \to\hat{\mathbb{C}}.
\end{align}

For a point $z_0\in \hat{\mathbb{C}}$, its backward orbit is defined by
\begin{align}
O^-(z_0)\coloneqq \bigcup_{n\ge 0}f^{(-n)}(z_0).
\end{align}
If the backward orbit is finite, the point $z_0$ is called an exceptional point, that is, 
\begin{align}
z_0\in E(f)\Longleftrightarrow O^-(z_0) \mathrm{\ is\ finife\  set.}
\end{align}
For a rational map of $d\ge 2$ there are at most two exceptional points, and all exceptional points lie in the Fatou set. Hence,
\begin{align}
z_0\in J(f)\Longrightarrow z_0\notin E(f).
\end{align}

Now let us fix the map to the DLM $f_a(z)$.
The anti-Julia set $J_{\mathrm{anti}}(f_a)$ of this map is defined by \eqref{antiJulia} as the closure of the set of the following;
\begin{align}
O_a^-(\infty)=\Big\{\left. f_a^{(-n)}(\infty)\right|\ n=0,1,2,3,\cdots\Big\}.
\end{align}

Since it is well known~\cite{Beardon1991} from the definition of the Julia set that 
\begin{align}
z_0\in J(f)\Longrightarrow \overline{O^-(z_0)}=J(f),
\label{heiho-}
\end{align}
then $\overline{O_a^-(\infty)}=J(f_a)$, if $\infty\in J(f_a)$,  and therefore the anti-Julia set $J_{\mathrm{anti}}(f_a)$ coincides with the Julia set itself.
On the other hand, if $\infty \notin J(f_a)$, $z_0=\infty$ is in the Fatou set, and consequently its backward orbit $O_a^-(\infty)$ lies in the Fatou set.  Therefore it is crucial if $z_0=\infty$ is an attracting point or repelling  point. If it is attracting the pre-anti-Julia set is separated from the pre-Julia set, has no chance to collide. Moreover we can show $J(f_a)\subsetneq \partial O_a^-(\infty)$, in this case.

To see the multiplier of
$f_a(z)$ at $\infty$, we define $w=1/z$ and
\begin{align}
F(w)=\frac{1}{f_a(1/w)}=\frac{w\big(w+\mu(1-a)\big)}{\mu(w-a)}.
\end{align}
We notice that $w=0$ is a fixed point if $a\ne 0$. The multiplier is
\begin{align}
\lambda_\infty=F'(0)=1-\frac{1}{a},
\end{align}
irrespective to $\mu$. Since $a$ is real, the fixed point $z_0=\infty$ is repelling for $a<1/2$, attracting for $a>1/2$, and neutral when $a=1/2$. From this result we finally see that the anti-Julia set coincides with the Julia set when $0<a\le 1/2$.

%\bibliographystyle{apsrev4-2}
%\bibliography{julia}% Produces the bibliography via BibTeX.


\end{document}